\documentclass[11pt,onecolumn,a4paper,accepted=2026-11-30]{quantumarticle}
\pdfoutput=1
\usepackage[utf8]{inputenc}
\usepackage[english]{babel}
\usepackage[T1]{fontenc}
\usepackage{amsmath}
\usepackage{amssymb}
\usepackage{amsthm}
\usepackage{graphicx}
\usepackage{mathrsfs}
\usepackage{braket}
\usepackage{comment}
\usepackage{hyperref}
\hypersetup{
  colorlinks,%
  citecolor=blue,
  filecolor=blue,%
  linkcolor=blue,%
  urlcolor=blue
}
\usepackage{color, xcolor, colortbl}
\usepackage[normalem]{ulem}
\usepackage{graphicx}
\usepackage{amsthm,amsmath,amssymb,amsfonts}
\usepackage{dcolumn}
\usepackage{url}
\usepackage{epstopdf}
\usepackage{enumitem}
\usepackage{algpseudocode}
\usepackage{bm}
\usepackage{appendix}
\usepackage{multirow}
\usepackage{braket}
\usepackage[english]{babel}
\usepackage{aliascnt}
\usepackage{hyperref}
\usepackage[capitalize]{cleveref}
\usepackage{algorithm}
\usepackage{xpatch}
\usepackage{tikz}
\usepackage{adjustbox}
\usepackage{xspace}
\usetikzlibrary{quantikz}
\usepackage{pifont}

\newcommand{\polylog}{\operatorname{poly}\log}
\newcommand{\poly}{\operatorname{poly}}

\newtheorem{thm}{Theorem}[section]
\crefname{thm}{Theorem}{Theorems}
\Crefname{thm}{Theorem}{Theorems}

\newaliascnt{lem}{thm}
\newaliascnt{rem}{thm}
\newaliascnt{prop}{thm}
\newaliascnt{cor}{thm}
\newaliascnt{assumption}{thm}
\newaliascnt{defn}{thm}
\newaliascnt{notation}{thm}
\newaliascnt{fact}{thm}
\newaliascnt{prob}{thm}

\newtheorem{lem}[lem]{Lemma}
\newtheorem{rem}[rem]{Remark}
\newtheorem{prop}[prop]{Proposition}

\newtheorem{assumption}[assumption]{Assumption}
\newtheorem{defn}[defn]{Definition}

\aliascntresetthe{lem}
\aliascntresetthe{rem}
\aliascntresetthe{prop}
\aliascntresetthe{cor}
\aliascntresetthe{assumption}
\aliascntresetthe{defn}
\aliascntresetthe{prob}
\aliascntresetthe{notation}
\aliascntresetthe{fact}

\crefname{lem}{Lemma}{Lemmas}
\Crefname{lem}{Lemma}{Lemmas}

\crefname{rem}{Remark}{Remarks}
\Crefname{rem}{Remark}{Remarks}

\crefname{prop}{Proposition}{Propositions}
\Crefname{prop}{Proposition}{Propositions}

\crefname{cor}{Corollary}{Corollaries}
\Crefname{cor}{Corollary}{Corollaries}

\crefname{assumption}{Assumption}{Assumptions}
\Crefname{assumption}{Assumption}{Assumptions}

\crefname{defn}{Definition}{Definitions}
\Crefname{defn}{Definition}{Definitions}

\crefname{prob}{Problem}{Problem}
\Crefname{prob}{Problem}{Problem}

\crefname{notation}{Notation}{Notations}
\Crefname{notation}{Notation}{Notations}

\crefname{fact}{Fact}{Facts}
\Crefname{fact}{Fact}{Facts}

\usepackage{bbm}

\NewDocumentCommand{\ketbra}{mG{#1}}{\mathinner{|{#1}\rangle\!\langle{#2}|}}

\usepackage{silence}


\usetikzlibrary{fit}
\tikzset{%
  highlight/.style={rectangle,rounded corners,fill=blue!15,draw,fill opacity=0.3,thick,inner sep=0pt}
}

\usepackage{tikz}
\usepackage{lipsum}
\usepackage{makecell}
\usepackage{multirow}
\usepackage[dvipsnames]{xcolor}
\usepackage{soul,xcolor}

\begin{document}
\title{Quantum Algorithms for Nonlinear Differential Equations via Pivot-Shifted Carleman Linearization}
\author{Ke Wang}
\thanks{These authors contributed equally to this work.}
\affiliation{Department of Mathematics, University of Michigan}
\author{Zikang Jia}
\thanks{These authors contributed equally to this work.}
\affiliation{Department of Mathematics, University of Michigan}
\author{Shravan Veerapaneni}
\affiliation{Department of Mathematics, University of Michigan}
\author{Zhiyan Ding}
\thanks{zyding@umich.edu}
\affiliation{Department of Mathematics, University of Michigan}
\maketitle
\begin{abstract}
We develop a pivot-shifted Carleman linearization framework for quantum algorithms solving quadratic nonlinear ordinary differential equations. By shifting the dynamics by a pivot state prior to Carleman lifting, and combining this with a Lyapunov transform and rescaling, we enlarge the class of nonlinear systems that can be efficiently simulated on quantum computers. For systems that exhibit stability in the shifted coordinates, we establish long time convergence of the truncated Carleman embedding. We prove that the truncation order scales only logarithmically with the simulation time and target precision, and we derive end-to-end quantum query complexity bounds for preparing a state proportional to the final solution.  {By introducing a modified nonlinearity condition, this framework relaxes the conventional lower bound requirement on the initial condition, thereby accommodating a broader class of initial states.} For more general systems that remain unstable after shifting, we provide short time convergence guarantees that are similarly free from the initial condition constraints. Numerical experiments on the logistic and the Lotka-Volterra equations demonstrate that an appropriate pivot choice improves stability and accuracy, and yields exponential error decay with truncation order. These results show that pivot shifting provides a practical and theoretically justified route for extending Carleman-based quantum algorithms to a broader class of nonlinear dynamical systems.
\end{abstract}

\section{Introduction}
Over the past several decades, quantum algorithms have attracted sustained interest because of their potential to provide significant speedups over classical methods~\cite{feynman1986quantum}. One of the most important quantum algorithmic primitives is the simulation of quantum systems, where the curse of dimensionality arises naturally and severely limits classical approaches. Since quantum dynamics are themselves governed by quantum mechanics, they are especially well suited for implementation on quantum computers. This structural compatibility has led to remarkable advances, including optimal-scaling Hamiltonian simulation~\cite{low2019hamiltonian,gilyen2019quantum} and near-optimal-scaling Lindblad simulation~\cite{cleve2016efficient,li2022simulating,ding2024simulating}. Building on these developments, recent works have begun to explore quantum advantages for more general classical computational tasks~\cite{lin2026quantum,dalzell2023quantum}, including linear system solvers~\cite{harrow2009quantum,childs2017quantum} and algorithms for linear differential equations~\cite{costa2022optimal,fang2023time,an2023linear,an2026quantum,berry2014high,berry2017quantum,morales2024quantum}, largely by exploiting the intrinsic linearity of quantum mechanics.

Although quantum computing offers powerful tools for simulating linear systems, many phenomena of scientific and practical interest are governed by nonlinear differential equations, including models arising in plasma physics, engineering, economics, and chemical reactions~\cite{hockney2021computer,scott2007nonlinear,black1973pricing,yu2018mathematical}. Their nonlinear structure makes accurate simulation substantially more challenging, even on modern classical high-performance computing platforms. On a quantum computer, an additional difficulty arises from the intrinsic linearity of quantum evolution, whereas the target dynamics are nonlinear. Therefore, a quantum algorithm cannot usually implement the nonlinear flow directly, and must instead represent it through an auxiliary linear framework or through intermediate operations that effectively reproduce the nonlinear behavior. This mismatch between linear quantum evolution and nonlinear target dynamics is the central obstacle in designing quantum algorithms for nonlinear differential equations.

Existing quantum approaches to nonlinear differential equations can largely be viewed as different strategies for overcoming the basic mismatch.
The Koopman--von Neumann framework embeds the dynamics into a linear operator evolution acting on observables or functions of the state, rather than directly on the state itself~\cite{joseph2020koopman}. Variational and hybrid quantum-classical methods offer a near term pathway by recasting nonlinear dynamics as optimization problems using parameterized quantum circuits, although the rigorous error analyses are not fully established~\cite{lubasch2020variational,pfeffer2022hybrid}.
Other early approaches similarly relied on more general linearization or approximation schemes to indirectly mimic nonlinear evolution on a quantum computer~\cite{leyton2008quantum}. Among these methods, Carleman linearization has emerged as one of the most promising frameworks. Its central idea is to lift the nonlinear system into a higher dimensional linear system by enlarging the state space and encoding higher order monomials of the solution~\cite{liu_efficient_2021,jennings_quantum_2025,kroviImprovedQuantumAlgorithms2023}. This construction is particularly attractive because it is systematic, compatible with quantum linear system techniques, and admits rigorous error analysis. Under dissipativity assumptions and a condition analogous to the Reynolds number in fluid dynamics, it yields efficient quantum simulation algorithms that remain stable for a long time~\cite{liu_efficient_2021}. Extensions beyond the dissipative regime have also been established for nonresonant systems~\cite{wuQuantumAlgorithmsNonlinear2025a,jennings_quantum_2025}, as well as for stable and conservative systems~\cite{jennings_quantum_2025}.

Despite this progress, the current applicability of Carleman linearization still relies on rather restrictive assumptions. These assumptions not only narrow the class of nonlinear systems that can be treated efficiently, but also impose nontrivial conditions on the initial value in order to guarantee convergence of the truncated Carleman expansion. Specifically, consider the $n$-dimensional quadratic ordinary differential equation
\begin{equation}\label{eq:qODEx}
\partial_t x = F_2 x^{\otimes 2} + F_1 x + F_0, \quad x(0) = x_{0},
\end{equation}
where $x \in \mathbb{R}^n$, $x^{\otimes 2} \in \mathbb{R}^{n^2}$, $F_2 \in \mathbb{R}^{n \times n^2}$, $F_1 \in \mathbb{R}^{n \times n}$, and $F_0 \in \mathbb{R}^n$. Existing analyses~\cite{jennings_quantum_2025,liu_efficient_2021,kroviImprovedQuantumAlgorithms2023,wuQuantumAlgorithmsNonlinear2025a} have established efficient quantum algorithms under different assumptions. In particular,~\cite{liu_efficient_2021,jennings_quantum_2025} establishes the convergence under three conditions. First, the linear part is required to be  {stable}, meaning its spectral abscissa satisfies $\alpha(F_1) {<} 0$. Second, a constraint similar to the Reynolds number is imposed,
\begin{equation}\label{eq:Rcondition}
R \coloneqq \frac{1}{-\alpha(F_1)}
\left(
\|x(0)\| \|F_2\|
+
\frac{\|F_0\|}{\|x(0)\|}
\right) < 1.
\end{equation}
This condition ensures that the dissipative effect of the linear term dominates both the nonlinear term and the forcing term, so that the norm of the solution remains nonincreasing. Furthermore, because $R$ scales inversely with the initial norm, $1/\|x(0)\|$, one needs a third condition, $\delta \leq \|x(0)\| \leq 1$ for some fixed $\delta > 0$.~ {Later works, such as~\cite{jennings_quantum_2025,kroviImprovedQuantumAlgorithms2023,wuQuantumAlgorithmsNonlinear2025a}, provide a more comprehensive discussion of the convergence conditions and extend the convergence to certain structured cases with $\alpha(F_1) = 0$.
}

Although these conditions are sufficient to guarantee efficient quantum algorithms, they are quite restrictive and describe a much narrower regime than the full class of nonlinear systems that are stable in the usual dynamical-systems sense. Many stable systems of practical interest fall outside this framework, and for such systems, establishing an end-to-end quantum advantage becomes substantially more challenging~\cite{dalzell2023quantum}. A simple example is the logistic equation $\partial_t x=x-x^2$. For every initial condition $x(0)>0$, the solution remains bounded and converges to a stable equilibrium. However, the linear coefficient is positive, and hence the condition $\alpha(F_1) {\leq}0$ fails. In this case, the standard Carleman approximation can diverge on long time scales and fail to achieve high accuracy, as demonstrated by our numerical experiments in~\cref{sec:example}.
It is therefore natural to ask whether one can develop provably efficient quantum algorithms that capture the dynamics of such systems. More generally, given the limitations of existing Carleman-based criteria, one may ask whether provably efficient quantum algorithms can be developed under weaker assumptions than those currently required. This leads to the first central question of our work:
\begin{quote}
\textit{Can we develop provably efficient algorithms under weaker  {stability} assumptions?}
\end{quote}

Beyond stable nonlinear systems, many systems of practical interest are unstable, and it is important to understand the scope of quantum algorithms in this broader setting. For unstable systems, solutions may grow rapidly, often exponentially in time, so one generally expects efficient simulation to be possible only over finite or short time intervals.
However,  {existing short-time convergence bounds for Carleman linearization often rely on assuming the origin is an equilibrium~\cite{forets2017explicit,amini2022carleman}, this setting does not cover general inhomogeneous quadratic systems, particularly when the relevant equilibrium lies away from the origin.}
This leads to the second central question of our work:

\begin{quote}
\textit{ {Can pivot shifting yield rigorous short-time convergence guarantees for general quadratic systems?}}
\end{quote}
In this work, we provide a modified Carleman linearization framework that addresses both questions above. Inspired by the classical idea of pivot rule in linear programming~\cite{terlaky1993pivot,bertsimas1997introduction}, we introduce a pivot-shifted Carleman linearization method. The main idea is to shift the dynamics by a pivot state $s$ and apply Carleman linearization to the shifted variable $u=x-s$, rather than directly to the original variable $x$. This produces a transformed quadratic system whose linear part can be substantially better behaved than that of the original equation.

We rigorously demonstrate that this modification significantly broadens the regime in which Carleman-based quantum algorithms can be analyzed. With a suitable choice of pivot, we establish long time convergence guarantees under assumptions that are strictly weaker and more flexible than the standard Carleman criteria. In particular, for a stable system that does not satisfy $\alpha(F_1)<0$, the pivot-shifted Carleman linearization can still yield an efficient quantum algorithm, provided that the pivot is chosen appropriately. The method remains useful even when the shifted system is still unstable. In this case, we prove short time convergence guarantees and show that pivot shifting can weaken the dependence on the initial condition while extending the time interval over which the Carleman approximation remains accurate.

Taken together, our results show both theoretically and numerically that pivot shifting plays two complementary roles. For stable systems, it enlarges the class of problems that admit efficient long time simulation. For unstable systems, it improves the effective short time approximation regime and allows the dynamics to be captured over longer time scales than standard Carleman linearization permits.

\subsection{Main results}\label{subsec:mainResults}
Our main results are obtained through a pivot-shifted Carleman linearization of the quadratic ordinary differential equation in~\cref{eq:qODEx}. For a chosen pivot $s$, we introduce the shifted variable  $u = x - s$, which satisfies
\begin{equation}\label{eqn:shiftedODE}
\partial_t u = F_{2,s} u^{\otimes 2} + F_{1,s} u + F_{0,s}.
\end{equation}
We then apply a Lyapunov transformation together with a rescaling by setting $v = Q u$, where $Q$ is chosen to control the stability and magnitude of the shifted dynamics. Carleman lifting is then applied to the transformed variable $v$. This pivot shifted formulation allows us to reduce the effective initial displacement and enforce a  {stability} condition on the transformed linear part.

We say that the differential equation becomes  {linearly} stable after pivot shifting if the following assumption holds:
\begin{assumption}[Stability after pivot shift]
\label{assumption:stability}
The spectral abscissa of $F_{1,s}$ is negative, namely
\begin{equation}\label{eq:Stability}
\alpha(F_{1,s}) = \max_i \mathrm{Re}(\lambda_i(F_{1,s})) < 0.
\end{equation}
\end{assumption}
By the Lyapunov theorem, ~\cref{assumption:stability} implies that there exists a positive definite matrix $P\succ 0$ such that
\begin{equation}\label{eqn:Lyapunov_P}
P F_{1,s} + F_{1,s}^\dagger P \prec 0.
\end{equation}
It reveals that the shifted linear dynamics is dissipative in the weighted inner product induced by $P$. Using the corresponding weighted norm $\|\cdot\|_P$ and logarithmic norm $\mu_P(\cdot)$, we impose the following condition.
\begin{assumption}[Weighted nonlinear stability condition]\label{assumption:stability2}
$\mu_P(F_{1,s})^2 - 4 \|F_{2,s}\|_P \|F_{0,s}\|_P > 0.$
\end{assumption}
This condition plays a role analogous to the requirement $R<1$ in~\cite{liu_efficient_2021,jennings_quantum_2025}.
Unlike the Reynolds-number-type condition in~\cref{eq:Rcondition}, this assumption does not place the initial state $u(0)$ in the denominator, and therefore avoids the singular behavior caused by the absence of a positive lower bound on the initial value.

\cref{assumption:stability}, together with the choice of the Lyapunov matrix $P$, makes the shifted linear dynamics strictly dissipative in the $P$ weighted norm, $\mu_P(F_{1,s})<0$.  \cref{assumption:stability2} controls the nonlinearity contribution by ensuring that the corresponding quadratic inequality has two distinct real roots. These assumptions together identify a regime in which the pivot-shifted Carleman system inherits a dissipativity structure in the transformed coordinates.
It is worth noting that the efficient implementation of the shifted Carleman linearization requires an efficient transformation between the original variable $x$ and the shifted variable $u=x-s$. This issue is discussed mainly in~\cref{sec:preliminaries}.
In this work, the computational cost of the algorithm is stated under the same cost model as in~\cite{liu_efficient_2021,jennings_quantum_2025}. Specifically, our complexity estimates assume an input model in which we have  {access} to block encodings of $F_2$, $F_1$, $F_0$, and $Q$, state preparation oracles for $\ket{x(0)}$ and $\ket{s}$, and the values of $\|x(0)\|$, $\|s\|$ and $\|u(T)\|$. In addition, we note that the pivot-shift idea may also be useful for marginally linearly stable systems~\cite{jennings_quantum_2025}. However, for simplicity and to highlight the main conceptual idea of the algorithm, we restrict our analysis to dissipative systems after the pivot shift, corresponding to the linearly stable shifted setting. Extending the framework to semi-dissipative systems is an interesting direction for future work.

We first state the main result for stable shifted systems, which establishes long time convergence of the pivot-shifted Carleman linearization and provides the corresponding quantum query complexity: 

\begin{thm}[Informal complexity for stable shifted systems]
\label{thm:stablecomplexityinformal}
Suppose Assumptions~\ref{assumption:stability} and~\ref{assumption:stability2} hold  {and $s\neq x(0)$}. Let $Q=\sqrt{P}/\gamma$ for a suitable rescaling parameter $\gamma>0$, and let $v=Q(x-s)$ be the transformed shifted variable. Then there exists a quantum algorithm that outputs an $\epsilon$-approximation of the final state $\ket{x(T)}$ using
\begin{equation}
\widetilde{\mathcal{O}}\left(
    \frac{
        \sqrt{\|x(T)-s\|^2+\|s\|^2}
    }{
        \|x(T)\|
    }
    \frac{
        \sqrt{\|x(0)\|^2+\|s\|^2}
    }{
        \|x(0)-s\|
    }
    \frac{
        \max_{t\in[0,T]}\|v(t)\|
    }{
        \|v(T)\|
    }
    \sqrt{T}
    \polylog(1/\epsilon)
\right)
\end{equation}
queries to the block encodings of $Q$ and $F_i$, and to the state-preparation oracles for $\ket{x(0)}$ and $\ket{s}$. Here, $\widetilde{\mathcal{O}}(\cdot)$ suppresses logarithmic factors and constants associated with the block encoding normalizations and condition parameters of the transformed system.
\end{thm}

\begin{rem}\label{re:stable}
    The choice of the pivot $s$ in the stable case is crucial for the efficiency of the quantum algorithm. In particular, $s$ should be chosen such that $F_{2,s},F_{1,s},F_{0,s}$ satisfies the  {linearly} stable condition in~\cref{assumption:stability}. The ideal choice of $s$ depends on the specific system. One useful guideline is to choose $s$ to be a stable equilibrium of the system, if such an equilibrium exists. For example, if the system has a locally exponentially stable equilibrium $x^\star$, then $F_{1,x^\star}:=F_1+F_2(x^\star\otimes I_n+I_n\otimes x^\star)$ has a negative spectral abscissa, meaning~\cref{assumption:stability} is satisfied with $s=x^\star$. In addition, $F_{0,x^\star}=F_0+F_1 x^\star + F_2 (x^\star)^{\otimes 2}=0$, so the condition in~\cref{assumption:stability2} is also satisfied if $\mu_P(F_{1,s})\neq 0$. This implies that, given a system that has a locally exponentially stable equilibrium $x^\star$, as long as $s$
    is chosen in a sufficiently small neighborhood of $x^\star$, the stability conditions above are satisfied and the Carleman linearization can be applied to yield an efficient quantum algorithm. This makes the pivot-shifted Carleman linearization a powerful tool for simulating stable nonlinear systems around their stable equilibria, even when the original system does not satisfy the standard Carleman criteria.

     {In addition, we note that when no prior knowledge of the system is available, identifying such an ideal choice may, in the worst case, require a long-time simulation of the dynamics. Nevertheless, there are physically relevant settings in which the equilibrium is approximately known and the primary interest lies in understanding the nonlinear convergence behavior toward equilibrium, such as Vlasov-Poisson equation~\cite{christensen2026}, Boltzmann equation~\cite{desvillettes2005trend}, Landau equation~\cite{guo2002landau}, etc.
    }
\end{rem}

The formal version of the above theorem is stated as~\cref{thm:main-stable} in~\cref{sec:complexity}. The proof consists of two main parts. First, we transform the original system into a shifted and rescaled system that satisfies the desired stability condition. Second, we apply Carleman linearization to the transformed system and analyze the resulting quantum algorithm. The first part involves constructing block encodings of the coefficient matrices of the shifted system and efficiently transforming the initial and final states between the original and shifted variables. These ingredients are discussed mainly in~\cref{sec:preliminaries}. The second part establishes convergence of the Carleman approximation for the shifted system and derives the corresponding quantum query complexity. In this part, we relax the assumptions on the initial value used in~\cite{jennings_quantum_2025}. In particular, instead of requiring $R_P = \frac{1}{-\mu_P(F_1)}\left( \|F_2\|_P \|x(0)\|_P + \frac{\|F_0\|_P}{\|x(0)\|_P}\right)<1$, our analysis can be applied to any $\|v(0)\|<1$. This removes the dependence on a lower bound for $\|x(0)\|_P$.  {This improvement is obtained by directly bounding the norm of the solution, rather than relying on the decay properties used in previous analyses. } In addition, while~\cite{jennings_quantum_2025} focuses on history state preparation, our analysis targets final state preparation and achieves a query complexity $\widetilde{\mathcal{O}}(\sqrt{T})$, improving over the best final state complexity $\widetilde{\mathcal{O}}(T)$  for Carleman linearization in~\cite{kroviImprovedQuantumAlgorithms2023}.

Although \cref{thm:main-stable} provides long time convergence guarantees for systems that become stable after pivot shifting, the choice of pivot is crucial and may require prior knowledge of the dynamics. For systems that are inherently unstable, or for which no practical pivot choice stabilizes the shifted dynamics, pivot-shifted Carleman linearization can still be useful for improving short time convergence behavior. In this regime, we choose the pivot to be the initial condition, $s=x(0)$, so that the shifted variable satisfies $u(0)=0$. This leads to the following short time convergence and query complexity guarantee.
\begin{thm}[Informal complexity for general shifted systems]
    \label{thm:unstablecomplexityinformal}
Choose $s=x(0)$ and $Q = I$,
then there exists $t^*>0$, depending on the coefficients $F_i$ and the initial condition $x(0)$, such that for any $T {<} t^*$, there is a quantum algorithm that outputs an $\epsilon$-approximation of the final state $\ket{x(T)}$ using
\[
        \widetilde{\mathcal{O}}\left(\frac{(1+\|x(0)\|)^2\sqrt{\|x(T)-x(0)\|^2+ \|x(0)\|^2 } }{\|x(T)\|}\frac{\max_{t\in[0,T]}\|x(t)-x(0)\|}{\|x(T)-x(0)\|\min\{\|x(T)-x(0)\|,\|x(T)\|\}}\frac{T^2}{\epsilon}\right)
\]
queries to the block encodings of $F_i$ and the state-preparation oracle for $\ket{x(0)}$.  Here, $\widetilde{\mathcal{O}}(\cdot)$ suppresses logarithmic factors and constants associated with the block encoding normalizations and condition parameters of the transformed system.
\end{thm}
The formal version of the above theorem is stated as~\cref{thm:main-unstable} in~\cref{sec:complexity}. We note that the above theorem does not impose any stability structure on the shifted system, and therefore applies to a much broader class of dynamics than standard Carleman linearization.  { Conventional Carleman linearization can converge for unstable systems over sufficiently short time intervals~\cite{forets2017explicit}. However, its local convergence bounds depend on the magnitude of the unshifted solution or initial condition. By choosing $s = x(0)$, the shifted variable $u = x-s$ satisfies $u(0) = 0$, so that the truncation error is instead controlled by the local displacement $\|x(t) - x(0)\|$. }
This can substantially improve the short time convergence behavior of the Carleman approximation, even when the underlying dynamics are unstable. The admissible short-time interval, however, may shrink when the shifted coefficients have a large overall scale.
The analysis of the above theorem builds on the quantum linear differential equation solver in~\cite{kroviImprovedQuantumAlgorithms2023,dutt2026simulating}. By slightly refining the upper bound on the  {norm of} matrix exponential of the Carleman linearized matrices, improved upon the estimation in~\cite[Lemma 16]{kroviImprovedQuantumAlgorithms2023}, we obtain a quadratic dependence on the evolution time $T$ for general dynamics. We note that, when $x(T)$ is very close to $x(0)$, the factor $1/(\|x(T)-x(0)\|\min\{\|x(T)-x(0)\|,\|x(T)\|\})$ in the complexity can become large. This scaling arises from the cost of the quantum linear solver applied to the Carleman-linearized system: when the target solution has norm close to zero, the success probability of preparing the normalized solution state can be small, and additional queries are needed to amplify this probability.

\subsection{Related works}

Several approaches have been developed for quantum algorithms for nonlinear differential equations. Early work by Leyton and Osborne~\cite{leyton2008quantum} incorporates quadratic nonlinearities through the amplitude of two copies of the solution state. The subsequent time-marching procedure requires multiple fresh copies of the solution state at each time step. Due to the no-cloning theorem, the resources required by this algorithm scale exponentially with time $T$ and the inverse precision $1/\epsilon$, although only  polylogarithmically with the dimension $n$ under suitable sparsity assumptions.

Numerous subsequent developments have sought to improve the dependence on time and precision, primarily for stable systems. For example, Carleman linearization has emerged as one of the most prominent frameworks: it systematically embeds nonlinear dynamics into a higher-dimensional linear system and admits rigorous error analysis. The quantum Carleman linearization approach was initially proposed by Liu et al.~\cite{liu_efficient_2021}, who showed that dissipative quadratic nonlinear differential equations satisfying a negative log-norm condition together with the Reynolds-type condition in~\cref{eq:Rcondition} can be solved efficiently using Carleman linearization, the quantum linear systems algorithm~(QLSA), and a first-order forward Euler method, with the complexity for final state preparation scales as $\widetilde{\mathcal O}(T^2/\epsilon)$. This was later improved by Krovi~\cite{kroviImprovedQuantumAlgorithms2023} to $\widetilde{\mathcal O}(T\polylog(1/\epsilon))$ by using a higher-order truncated Taylor-series solver for the resulting linear system. 

More recently, Wu et al.~\cite{wuQuantumAlgorithmsNonlinear2025a} extended the applicable regime to certain nonresonant systems through a refined spectral analysis of the Carleman matrix, while Jennings et al.~\cite{jennings_quantum_2025} replaced the negative log-norm assumption with the more general stability condition in~\cref{eq:Stability} by introducing a Lyapunov transform prior to the Carleman lifting. Unlike previous algorithms that output an approximation of the final state $\ket{x(T)}$,
Jennings et al.~\cite{jennings_quantum_2025} achieve an  $\mathcal{O}(\poly\log(T/\epsilon))$ scaling for preparing an approximation to the history state proportional to $[x(0); x(h); \cdots; x(Mh)]$. The history state preparation is naturally based on the "all at once" formulation, whereas extracting the final state requires an additional padding and amplitude amplification procedure to boost the success probability of measuring the final state. In addition, all of these works focus on dissipative,  {linearly} stable or marginally  {linearly} stable dynamics, and do not provide theoretical guarantees for unstable ODE systems where $\alpha(F_1)>0$. In contrast, our paper develops a modified Carleman quantum algorithm for unstable systems.  {Our proof of short-time convergence is related to~\cite{forets2017explicit}, where the authors derived explicit short-time error bounds for the Carleman linearization of homogeneous polynomial differential equations using both backward integration and power series techniques. In our setting, since the pivot-shift point is not necessarily an equilibrium point, the shifted system is generally inhomogeneous. We therefore adapt their result to the inhomogeneous setting and also rigorously discuss the complexity of quantum implementation.} 

Related developments have also shown how boundary conditions can be incorporated into quantum differential equation algorithms through explicit block encodings~\cite{kharazi2025explicit} or penalty projection methods~\cite{schleich2025arbitrary}. In addition, Carleman-type constructions have been extended from quadratic ordinary differential equations to polynomial nonlinear systems of arbitrary order~\cite{surana_efficient_2023}.
Beyond Carleman linearization, other frameworks have also been explored for nonlinear dynamics on quantum computers. These include the Koopman--von Neumann approach~\cite{joseph2020koopman}, which reformulates the dynamics as a linear evolution on observables and allows flexible choices of basis functions such as Fourier expansions~\cite{katz2025efficient} and extensions based on coupled bosonic modes~\cite{gan2025provably}; the homotopy perturbation method~\cite{xue2021quantum}, which reduces the problem to solving a sequence of linearized equations; and quantum algorithms for nonlinear Schr\"odinger-type equations~\cite{lloyd2020quantum}.

Our work is related to recent efforts that introduce shifting or piecewise embedding ideas into Carleman linearization. Endo and Takahashi~\cite{endo2024divergence} proposed a \emph{pivot-switching} method, in which the pivot state $s$ is repeatedly updated as the trajectory evolves, while the Carleman lifting is still performed about the origin. We emphasize that, despite the similar nomenclature, their method differs from ours in two essential respects: they switch pivots dynamically over the simulation, whereas the present work uses a single static shift; and their generator expansion is performed only locally, whereas we Carleman-lift directly around the shifted pivot. Although the strategy of~\cite{endo2024divergence} can extend the accessible simulation time, it generally produces a linear approximation with a non-sparse coefficient matrix, which results in a substantial increase in the cost of block encoding and quantum simulation. In addition, during the pivot-switching step, the Carleman linearization requires block encodings of both the new coefficient matrix and the updated initial state. To the best of our knowledge, it remains unclear whether these objects can be prepared efficiently.
Closely related ideas were explored in the global piecewise Carleman embedding framework of~\cite{novikau2025globalizing}, which showed that some chaotic systems can be approximated by Carleman linearization if pivot shifts are performed sufficiently frequently. Shifted formulations also appear in quantum algorithms based on lattice Boltzmann methods~\cite{jenningsEndtoendQuantumAlgorithm2025}. In the quantum setting, however, repeated pivot updates can introduce substantial overhead, since they require measurements to estimate both the current state and the next pivot, while the resulting loss of sparsity may significantly increase the cost of block encoding. In our work, we mainly focus on a single pivot shift, which avoids the overhead of repeated measurements and pivot updates. To the best of our knowledge, it remains unclear how to implement pivot switching efficiently in practice. Although we do not consider switching between pivot points, our results establish a theoretical foundation for the convergence of a single step in~\cite{endo2024divergence,novikau2025globalizing}. For unstable systems, our analysis further provides a theoretical characterization of an upper bound on the switching time. We will discuss the potential of multiple pivot shifts (pivot switching) in~\cref{sec:discussion}.  {Very recently,
\cite{christensen2026} studied Carleman linearization of a shifted Vlasov-Poisson System using a Fourier-Hermite Expansion. They chose the Maxwellian equilibrium as the pivot $s$ and observed improved convergence. This provides a concrete and physically motivated example of the pivot-shift strategy considered in our work.}

The key primitive underlying the Carleman-based algorithm is an efficient quantum solver for linear differential equations. Existing quantum algorithms for linear differential equations are typically based on the QLSA~\cite{berry2017quantum}, time-marching methods~\cite{fang2023time}, or reformulations in terms of Hamiltonian simulation~\cite{an2023linear,PhysRevLett.133.230602,an2026quantum}. The optimal query complexity for general linear differential equations is limited by the $\Omega(T)$ lower bound for Hamiltonian simulation~\cite{berry2015hamiltonian}. In contrast, recent works have shown that for certain classes of linear differential equations, we can obtain fast forwarded complexity~\cite{an2025quantum}.
In particular, for stable linear systems, the complexity of  history state preparation can be reduced to $\mathcal{O}(\sqrt{T})$~\cite{jennings2024cost}.
For dissipative linear systems, the query complexity of final state preparation can be improved to $\widetilde{\mathcal{O}}(\sqrt{T})$ via QLSA~\cite{An_2026}, and the dependence on $T$ can be eliminated using linear combination of Hamiltonian simulation(LCHS) or time-marching approaches~\cite{yangQuantumDifferentialEquation2025}. In this work,
we analyze the QLSA-based Carleman linearization with the linear differential equation solver developed in~\cite{An_2026}.
Our analysis is tailored to final state preparation and achieves a query complexity $\widetilde{\mathcal{O}}(\sqrt{T})$ using~\cite{An_2026}. This improves over the  final state preparation cost obtained from the linear differential equation solvers used in~\cite{jennings_quantum_2025}, namely $\widetilde{\mathcal{O}}(T^{3/4})$ scaling based query complexity on~\cite{jennings2024cost}. Moreover, we formulate the  query complexity in terms of the condition number $\kappa_Q$, rather than the norm of $Q$, making the bound invariant under rescaling of $Q$. Such rescaling invariance is essential for a consistent treatment of block encoding normalization factors.

\begin{table}[t]
    \renewcommand{\arraystretch}{1.8}
    \centering
    \adjustbox{max width = \linewidth}{
    \begin{tabular}{c|c|c|c|c|c}
        \hline\hline
        \multirow{2}{8em}{\makecell[c]{\textbf{}}}
        & \multirow{2}{10em}{\makecell[c]{\textbf{Method}}}
        & \multirow{2}{10em}{\makecell[c]{\textbf{Stability}}}
        & \multicolumn{2}{c|}{\textbf{Query complexity}}
        & \multirow{2}{9em}{\makecell[c]{\textbf{Remark}}} \\
        \cline{4-5}
        & & & \textbf{History state} & \textbf{Final state} & \\
        \hline
        \multirow{2}{8em}{\makecell{Linear ODEs}}
        & \makecell{Truncated Taylor ~\cite{kroviImprovedQuantumAlgorithms2023,berry2014high}}
        & $\max_t \|e^{At}\| = \mathcal{O}(1)$
        & $\widetilde{\mathcal{O}}\left(T\,\mathrm{poly}\log(1/\epsilon)\right)$
        & $\widetilde{\mathcal{O}}\left(T\,\mathrm{poly}\log(1/\epsilon)\right)$
        & QLSA based\\
        \cline{2-6}
        & \makecell{LCHS/Schr\"odingerization~\cite{an2023linear,PhysRevLett.133.230602}}
        & $A(t)+A(t)^\dag<0$
        & $\widetilde{\mathcal{O}}\left(T\,\mathrm{poly}\log(1/\epsilon)\right)$
        & $\widetilde{\mathcal{O}}\left(T\,\mathrm{poly}\log(1/\epsilon)\right)$
        & Hamiltonian simulation\\
        \hline
        \multirow{3}{8em}{\makecell{Fast-forwarded\\ Linear ODEs}}
        & \makecell{Fast-forwarded Taylor~\cite{jennings2024cost}}
        & $\alpha < 0$
        & $\widetilde{\mathcal{O}}\left(T^{1/2}(\log(1/\epsilon))^2\right)$
        & $\widetilde{\mathcal{O}}\left(T^{3/4}(\log(1/\epsilon))^2\right)$
        & QLSA based \\
        \cline{2-6}
        & \makecell{Fast-forwarded Euler/Dyson~\cite{An_2026}}
        & $A(t)+A(t)^\dagger <0$
        & $\mathcal{O}\left(\log(T)(\log(1/\epsilon))^2\right)$
        & $\widetilde{\mathcal{O}}\left(T^{1/2}(\log(1/\epsilon))^2\right)$
        & QLSA based \\
        \cline{2-6}
        & \makecell{Fast-forwarded LCHS~\cite{yangQuantumDifferentialEquation2025}}
        & \makecell{$A(t)+A(t)^\dagger <-2\eta<0$\\$T=\widetilde{\Omega}(\frac{1}{\eta})$}
        & $\mathcal{O}(\log^3(1/\epsilon))$
        & $\mathcal{O}(\log^3(1/\epsilon))$
        & Hamiltonian simulation\\
        \hline
        \multirow{4}{8em}{\makecell{Carleman-based\\ Nonlinear ODEs}}
        & Liu et al.~\cite{liu_efficient_2021}
        & $\mu<0, R<1$
        & /
        & $\widetilde{\mathcal O}(T^2/\epsilon)$
        & Forward Euler\\
        \cline{2-6}
        & Krovi~\cite{kroviImprovedQuantumAlgorithms2023}
        & $\mu<0, R<1$
        & /
        & $\widetilde{\mathcal O}(T\polylog(1/\epsilon))$
        & \makecell{Truncated Taylor} \\
        \cline{2-6}
        & Jennings et al.~\cite{jennings_quantum_2025}
        & \makecell{$\alpha <0, R_P<1$}
        & $\mathcal{O}(\poly\log(T/\epsilon))$
        & /
        & \makecell{Fast-forwarded Taylor~\cite{jennings2024cost,An_2026}} \\
        \cline{2-6}
        & This work
        & \cref{assumption:stability,assumption:stability2}
        & /
        & $\widetilde{\mathcal{O}}(\sqrt{T}\polylog(1/\epsilon))$
        & \makecell{Fast-forwarded Taylor~\cite{An_2026}} \\
        \hline\hline
    \end{tabular}
    }
    \caption{Comparison of quantum algorithms for linear and nonlinear differential equations. Here, \(\alpha\) denotes the spectral abscissa and \(\mu\) denotes the logarithmic norm of the linear part of the differential equation. The table focuses on the dissipative or stable setting and the time-independent case.}
    \label{tab:comparison_ode_algorithms}
\end{table}

\subsection{Organization}
The remainder of the paper is organized as follows. In~\cref{sec:preliminaries}, we introduce the notation, norm properties and quantum primitives.
In~\cref{sec:algorithm}, we present the pivot-shifted Carleman linearization algorithm.
In~\cref{sec:complexity}, we analyze the error convergence and complexity.  In~\cref{sec:example}, we demonstrate the performance of our algorithm through numerical experiments. We summarize our results and discuss future directions in~\cref{sec:discussion}.

\section{Preliminaries}\label{sec:preliminaries}

\paragraph{Norm and Properties}
Throughout this paper, we use $\|\cdot \|$ to denote the Euclidean vector norm or the induced matrix norm. In this section, we introduce the property of the weighted norm for matrices in $\mathbb{R}^{n\times 1}$, $\mathbb{R}^{n\times n}$ and $\mathbb{R}^{n\times n^2}$. Given a vector $\psi\in\mathbb{R}^n$, we denote $\ket{\psi} = \psi/\|\psi\|$ as the normalized quantum state.

Given a positive definite matrix $O \succ 0 $, we define the $O$-induced inner product by
\begin{equation}
     \langle x, y \rangle_O = x^\dag O y.
\end{equation}
Then,
\begin{itemize}
    \item For $x\in \mathbb{R}^{n\times 1}$, the weighted norm  is
    \begin{equation}
        \|x\|_O = \sqrt{\langle x, x\rangle_O}.
    \end{equation}
    \item For $F \in \mathbb{R}^{n\times n }$, the weighted matrix norm is
    \begin{equation}
        \|F\|_O := \max_{x\neq 0}\frac{\|Fx\|_O}{\|x\|_O} = \max_{x\neq 0}\frac{\|\sqrt{O}Fx\|}{\|\sqrt{O}x\|} = \max_{y\neq 0}\frac{\left\|\sqrt{O}F\sqrt{O}^{-1}y\right\|}{\|y\|} = \left\|\sqrt{O}F\sqrt{O}^{-1}\right\|\,.
    \end{equation}
    \item For $G \in\mathbb{R}^{n\times n^2}$, the norm follows
    \begin{equation}
 \|G\|_O :=\max_{x\neq 0}\frac{\|Gx\|_O}{\|x\|_{O\otimes O}} = \| O^{1/2} G (O^{-1/2}\otimes O^{-1/2})\|\,.
\end{equation}
\end{itemize}

The logarithmic norm can also be defined from the vector norm. For a matrix $F \in \mathbb{R}^{n\times n }$, its log norm is defined as
\begin{equation}\label{eq:lognorm}
    \mu(F) = \frac{1}{2}\max_i\lambda_i\left(F + F^\dag\right)\,.
\end{equation}
 {Here, $\lambda_i(A)$ denotes the eigenvalues of the matrix $A$.}
Correspondingly, the weighted log norm follows
\begin{equation}
    \mu_O(F) = \max_{x\neq 0} \mathrm{Re}\left(\frac{\langle F x, x\rangle_O}{\langle x, x\rangle_O}\right)=\lim_{h\rightarrow 0^+}\frac{\|I + hF\|_O-1}{h}=\lim_{h\rightarrow 0^+}\frac{\left\|I + h\sqrt{O}F\sqrt{O}^{-1}\right\|-1}{h} = \mu\left(\sqrt{O}F\sqrt{O}^{-1}\right)\,.
\end{equation}

For two functions $f, g$ that characterize complexity scaling, we use the notation $f = \mathcal{O}(g)$ if there exists some constant $c>0$ such that $f\leq c g$. We write $f = \Theta(g)$ or $f\sim g$ if both $f = \mathcal{O}(g)$ and $g = \mathcal{O}(f)$ hold. Additionally, $f = \widetilde{\mathcal{O}}(g)$ denotes the polylogarithmic factors are suppressed, i.e., $ f = \mathcal{O}(g\polylog (g))$.

\paragraph{Quantum primitive} Quantum hardware naturally implements unitary evolution. However, most matrices in scientific computing do not satisfy the unitary property. A natural idea is therefore to construct a larger unitary matrix so that the nonunitary matrix $A$ becomes the upper left block, i.e.
\begin{equation}
U_A \approx
    \begin{pmatrix}
        A/\alpha & *\\ * & * \\
    \end{pmatrix}.
\end{equation}
This is the key idea behind the block encoding, an important building block in quantum algorithm design.
\begin{defn}[Block Encoding]
Let $A\in\mathbb{C}^{2^n\times  {2^m}}$ and $q\geq \max\{n,m\}$. Suppose that there exists $\alpha, \epsilon>0$ and $q$-qubit unitary $U_A$ such that
    \begin{equation}
        \| A- \alpha (\bra{0}^{\otimes  {(q-n)}}\otimes I_{2^n}) U_A (\ket{0}^{\otimes {(q-m)}}\otimes I_{2^m})\|\leq \epsilon.
    \end{equation}
    We refer to $U_A$ as a generalized block encoding of $A$.
     {When $n  = m$, and $a = q-n$,}  $U_A$ is called an $(\alpha, a, \epsilon)$-block encoding of $A$. The block encoding is exact when $\epsilon = 0$.
\end{defn}
 The block-encoding framework is not restricted to square matrices and extends naturally to rectangle matrices.
While block encoding is not guaranteed to be efficient for arbitrary dense matrices, linear combinations of unitaries (LCU) provides flexibility in constructing suitable block encodings~\cite[Lemma 52]{gilyen2019quantum}.
We introduce the slightly generalized version with different normalizing factors stated in~\cite[Lemma 5]{li2023succinct}.
\begin{lem}[Linear combination of unitaries]
    Suppose $A \coloneqq \sum_{j=1}^{m} y_j A_j \in \mathbb{C}^{2^n\times 2^n}$, where $A_j \in \mathbb{C}^{2^n \times 2^n}$ and $y_j>0$ for all $j \in \{1, \ldots m\}$. Let $U_j$ be an $(\alpha_j, a, \epsilon)$-block-encoding of $A_j$ and $ {\omega} = \sum_{j = 1}^m y_j \alpha_j$. Define the state preparation oracle pair as
    \begin{equation}
    \begin{aligned}
        \mathbf{Prep}: \ket{0}\rightarrow \sum_{j = 0}^{m-1}\sqrt{\frac{\alpha_jy_j}{ {\omega}}}\ket{j},\\
    \end{aligned}
    \end{equation}
    and the select oracle as
    \begin{equation}
        \mathbf{Select} = \sum_{j = 0}^{m-1}\ketbra{j}\otimes U_j.
    \end{equation}
    Then the matrix
    \begin{equation}
        (\mathbf{Prep}^\dag\otimes I ) \mathbf{Select} (\mathbf{Prep}\otimes I )
    \end{equation}
    is an $\left(\sum_{j = 1}^{m} \alpha_jy_j, a+\log(m),\sum_{j = 0}^{m-1}\alpha_jy_j\epsilon\right)$-block encoding of $A$.
\end{lem}
Additionally, the matrix inverse can be implemented through the QLSA with the following query complexity.
\begin{lem}\label{lem:QLSP}(\cite[Theorem 11]{costa2022optimal})
    Suppose a matrix $A\in\mathbb{C}^{2^n\times 2^n}$ with $\|A\| = 1$ and $\|A^{-1}\| = \kappa$. Given a block encoding oracle of the matrix $A$ and a state preparation oracle for $\ket{b}$, then the approximation of the normalized state $\ket{A^{-1}b}$ with precision $\epsilon$ can be produced with
    \begin{equation}
        \mathcal{O}(\kappa \log(1/\epsilon))
    \end{equation}
    calls to the oracles.
\end{lem}

\paragraph{Input model}
To ensure efficient encoding of the problem in~\cref{eq:qODEx} onto quantum computer, we make the following assumptions on the input model:
\begin{itemize}
  \item We assume
   {$U_{F_i}$} for $i = 0,1,2$ are the block-encoding of the matrix $F_i$  {with the subnormalization coefficient $\alpha_{F_i}$}  if $U_{F_i}$ is a unitary and
  \begin{equation}
         { \|F_i - \alpha_{F_i} (\bra{0}\otimes I) U_{F_i} (\ket{0}\otimes I)  \| \leq \epsilon, \quad i = 0,1,2.}
  \end{equation}
  \item We assume the matrix $Q$ that is proportional to $\sqrt{P}$  {is} block encoded by  {$U_{Q}$} with normalizing factor $\alpha_Q=\mathcal{O}\left(\|Q\|\right)$ so that
  \begin{equation}
       {\|Q - \alpha_Q (\bra{0}\otimes I ) U_Q (\ket{0}\otimes I ) \| \leq \epsilon.}
  \end{equation} We denote the condition number of $Q$ by $\kappa_Q = \|Q\|\|Q^{-1}\|$.
  \item We assume $\|x_0\|$ is  {known} and the normalized vector $\ket{x_0}$ can be efficiently prepared by an initial state preparation oracle, meaning  {we are given a unitary oracle $O_{x_0}$ so that }
  \begin{equation}
       {O_{x_0}\ket{0} =\ket{x_0}.}
  \end{equation}
  \item Given the pivot $s$, we assume $\|s\|$ is known and a normalized state $\ket{s} = s/\|s\|$ can be efficiently prepared by an initial state preparation oracle $O_s$  {so that}
  \begin{equation}
       {O_s\ket{0} = \ket{s}.}
  \end{equation}
   { As discussed in Remark I.4, when prior information about the system dynamics is available, a stable equilibrium point provides a natural choice of the pivot $s$ for establishing long time convergence of the Carleman linearization.}
  \item  {We assume $\|u(T)\|$ is known and also denote the differential equation solver for the shifted final state $u(T)$  by $O_{u(T)}$, so that for success probability $p$, it satisfies
        \[
        O_{u(T)}\ket{0^{m+k}} = \sqrt{p}\ket{0^k}\ket{u(T)} + \sqrt{1-p}\ket{\perp}, \quad \bra{0^k}\ket{\perp} = 0.
        \]
    Specifically, $O_{u(T)}$ will be constructed from the Carleman linearization of the shifted system. The detailed discussion and estimation of $p$ can be found in  {Appendix}~\ref{appd:QuantumOracles}.}
\end{itemize}

Given a pivot state $s$ and the shifted variable $u = x - s$,~\cref{eq:qODEx} can be rewritten as~\cref{eqn:shiftedODE} with
\begin{equation}\label{eq:Fis}
    F_{2, s} = F_{2}, \quad F_{1, s}= F_1 + F_2(s\otimes I_n + I_n\otimes s), \quad F_{0,s} = F_2s^{\otimes 2}+ F_1 s + F_0\,.
\end{equation}
The cost of constructing $F_{2,s}$, $F_{1,s}$, and $F_{0,s}$ and the state preparation oracle for the shifted initial state $\ket{u(0)}$ are summarized in the following lemma.
\begin{lem}\label{lem:query0} Assume access to the block encoding oracle  {$U_{F_2}, U_{F_1}, U_{F_0}$}, the state preparation oracles $ {O_{x_0}}, O_s$, as introduced above. The following hold:
    \begin{itemize}
        \item The block encoding oracle for $F_{0,s},F_{1,s},F_{2,s}$ can be constructed with normalizing factor $\alpha_{F_{2,s}} = \alpha_{F_2}$, $\alpha_{F_{1,s}} = \alpha_{F_1}+2\|s\| \alpha_{F_2}$ and $\alpha_{F_{0,s}} = \alpha_{F_0}+\alpha_{F_1} \|s\| + \alpha_{F_2}\|s\|^2$.
        \item  {The block encoding oracle for $E_{0,s}, E_{1,s}, E_{2,s}$ can be constructed with normalization factor $\alpha_{E_{0,s}} = \alpha_{F_{0,s}}\alpha_Q$, $\alpha_{E_{1,s}} = \alpha_{F_{1,s}}\kappa_Q$ and $\alpha_{E_{2,s}} = \alpha_{F_{2,s}}\kappa^2_Q/\alpha_Q$. }
        \item  {For $s \neq x(0)$,} the shifted initial state $\ket{u(0)} = \ket{x(0)-s}$ can be prepared with $\mathcal{O}\left(\frac{\sqrt{\|x(0)\|^2 + \|s\|^2}}{\|x(0)-s\|}\right)$ queries to $O_{x_0}$ and $O_s$.  {In the special case when $s = x_0$, the shifted initial value $u(0) = 0$, and therefore no normalized shifted state needs to be prepared. }
        \item The final state $\ket{x(T)}$ can be shifted back from $\ket{u(T)}$ with $\mathcal{O}\left(\frac{\sqrt{\|x(T)-s\|^2+ \|s\|^2 } }{\sqrt{p}\|x(T)\|}\right)$  queries to the preparation oracles $O_{u(T)}$ and $O_s$.
    \end{itemize}
\end{lem}
The detailed construction and cost of each oracle can be found in  Appendix~\ref{appd:QuantumOracles}.

\section{Details of the Pivot-Shifted Algorithm}\label{sec:algorithm}
\begin{figure}
    \centering
    \begin{adjustbox}{max width = \linewidth}
    \begin{tikzpicture}[
        node distance=11mm and 20mm,
        box/.style    ={rectangle, rounded corners=4pt,
                        draw=blue!55!cyan!75!black, line width=0.7pt,
                        fill=blue!7, minimum height=10mm, align=center, inner sep=4pt},
        carleman/.style={rectangle, rounded corners=4pt,
                        draw=orange!75!red!70!black, line width=0.7pt,
                        fill=orange!12, minimum height=10mm, align=center, inner sep=4pt},
        lbl/.style    ={inner sep=2pt},
        slbl/.style   ={inner sep=2pt, font=\small},
        ar/.style     ={-{Latex[length=2mm,width=1.8mm]}, line width=0.7pt,
                        draw=blue!55!cyan!75!black},
    ]
        \node[box] (Ax) {$\partial_t x = F_2 x^{\otimes 2} + F_1 x + F_0$};
        \node[box, below=of Ax] (Au) {$\partial_t u = F_{2,s} u^{\otimes 2} + F_{1,s} u + F_{0,s}$};
        \node[box, below=of Au] (Av) {$\partial_t v = E_{2,s} v^{\otimes 2} + E_{1,s} v + E_{0,s}$};

        \draw[ar] (Ax) -- node[lbl, left=2pt] {Pivot shift $u = x - s$} (Au);
        \draw[ar] (Au) -- node[lbl, left=2pt] {Lyapunov transform $v = Qu$} (Av);

        \node[carleman, right=of Av] (Carl) {$\partial_t \hat{z} = B_N \hat{z} + d_N$};
        \node[lbl, below=1mm of Carl] {Carleman lifting};

        \node[above=14mm of Carl, font=\small] (BN)
            {$B_N = \begin{pmatrix}
                B_{1,1} & B_{1,2} &        & 0          \\
                B_{2,1} & B_{2,2} & \ddots &            \\
                        & \ddots  & \ddots & B_{N-1,N}  \\
                0       &         & B_{N,N-1} & B_{N,N}
            \end{pmatrix}$};

        \draw[ar] ([yshift= 2mm]Av.east) to[bend left=22]
            node[lbl, sloped, above] {stable}   ([yshift= 2mm]Carl.west);
        \draw[ar] ([yshift=-2mm]Av.east) to[bend right=22]
            node[lbl, sloped, below] {unstable} ([yshift=-2mm]Carl.west);

        \node[box, right=16mm of Carl] (QLSA) {$A_{M, M_p - 1} Y = b_{M, M_p - 1}$};
        \node[lbl, below=1mm of QLSA] {QLSA};

        \draw[ar] (Carl) -- (QLSA);

        \node[lbl, above=6mm of QLSA] (ztilde) {$\ket{\tilde{z}}$};
        \node[lbl, above=11mm of ztilde] (zone)  {$\ket{\tilde{z}_1}$};
        \node[lbl, above=11mm of zone]   (xT)    {$\ket{x_T}$};
        \draw[ar] (QLSA) -- (ztilde);
        \draw[ar] (ztilde) --
            node[slbl, right=1pt] {Amp.\ amplification} (zone);
        \draw[ar] (zone) --
            node[slbl, right=1pt, align=left] {$Q^{-1}\;+\;$\\Shift oracle} (xT);

        \begin{scope}[on background layer]
            \fill[black!5, rounded corners=6pt]
                ($(current bounding box.south west) + (-3mm, -2mm)$)
                rectangle
                ($(current bounding box.north east) + ( 3mm,  2mm)$);
        \end{scope}
    \end{tikzpicture}
\end{adjustbox}
    \caption{Flowchart for the pivot-shifted Carleman linearization algorithm.}
    \label{fig:flowchart}
\end{figure}
We present the pivot-shifted Carleman linearization, incorporating the Lyapunov transform in~\cite{jennings_quantum_2025}, for the quadratic nonlinear differential equation in~\cref{eq:qODEx}.
Our approach applies the Carleman linearization directly around the pivot $s$ after expanding the generator at the same point. This formulation leads to a sparser lifting matrix,  is compatible with single-shot measurement and allows us to combine it with the quantum algorithms designed for stable nonlinear differential equations. The algorithm consists of four steps as illustrated in~\cref{fig:flowchart}:
\begin{enumerate}[label=(\arabic*)]
    \item \textbf{Pivot shift.}
Given a suitable pivot vector $s\in\mathbb{R}^n$, we apply a pivot-shift transformation $u = x-s$, with initial condition $u(0)=x(0)-s$. Under this change of variables, \cref{eq:qODEx} can be rewritten as
\begin{equation}\label{eq:MPSquadratic}
    \partial_t u = F_{2, s} u^{\otimes 2} +   F_{1, s} u +  F_{0,s}\,,
\end{equation}
where the shifted coefficient matrices are given by~\cref{eq:Fis}. The block encoding oracles for $F_{2,s}$, $F_{1,s}$, and $F_{0,s}$, as well as the state preparation oracle for the shifted initial state $u(0)$, can be constructed efficiently, as summarized in~\cref{lem:query0}. When the prior information about a stable region is available, the pivot $s$ can be chosen in that region to improve the stability of the shifted dynamics. In the absence of such information, or when the system is unstable, a natural choice is $s = x(0)$, as in~\cite{novikau2025globalizing}. This choice does not generally stabilize the system, but it gives $u(0) = 0$ and supports short time convergence of the Carleman approximation.

\item  \textbf{Lyapunov Transform.}
With a \textit{positive definite matrix} $Q$, we then consider the change of variable
\begin{equation}
\label{eq:LyapunovTransform}
    v = Qu.
\end{equation}
Under this transformation, the differential equation becomes
    \begin{equation}\label{eqn:v_t_equation}
        \partial_t v = E_{2,s} v^{\otimes 2} + E_{1,s} v+E_{0,s}  \,,
    \end{equation}
    where the coefficient matrices are given by
    \begin{equation}\label{eq:EfromF}
        E_{2,s} = Q F_{2,s} (Q^{-1}\otimes Q^{-1}), E_{1,s} = Q F_{1,s} Q^{-1}, E_{0,s}  = Q F_{0,s}, v_0 = Qu_0\,.
    \end{equation}
This transform is introduced to enhance the stability properties of the linearized system. Under \cref{assumption:stability,assumption:stability2}, we choose $Q= \sqrt{P}/\gamma$, following~\cite{jennings_quantum_2025}, so that the linearized differential equation is dissipative and that can be solved efficiently by a quantum linear differential equation solver. For general systems where the prior knowledge of the underlying system might not be available, we also provide analysis when $Q = I$.

\item  \textbf{Carleman lifting.} We truncate the Carleman linearization at order $N$ and denote the resulting lifted solution as $\hat z=(\hat z_k)_{k=1}^N$ such that
\begin{equation}\label{eq:carleman-trunc-LM}
\frac{d}{dt} \hat{z} = \frac{d}{dt}
\begin{bmatrix}
\hat z_1 \\ \hat z_2 \\ \vdots \\ \hat z_N
\end{bmatrix}
=
\begin{bmatrix}
 B_{1,1} & B_{1,2} & 0 & \cdots & 0 \\
 B_{2,1} & B_{2,2} & B_{2,3} & \ddots & \vdots \\
 0 & B_{3,2} & B_{3,3} & \ddots & 0 \\
 \vdots & \ddots & \ddots & \ddots & B_{N-1,N} \\
 0 & \cdots & 0 & B_{N,N-1} & B_{N,N}
\end{bmatrix}
\begin{bmatrix}
\hat z_1 \\ \hat z_2 \\ \vdots \\ \hat z_N
\end{bmatrix}
+
\begin{bmatrix}
 E_{0,s} \\ 0 \\ \vdots \\ 0
\end{bmatrix}
=: B_N \hat z + d_N .
\end{equation}
The initial condition approximates the lifted exact solution
\[
\hat z (0) \approx \begin{bmatrix}
        v_0; v_0^{\otimes 2}; \cdots ; v_0^{\otimes N}
    \end{bmatrix}.
\]
This block tridiagonal structure is inherited from the quadratic form of the differential equation.
For admissible indices $j$ from $1$ to $N$, the blocks are given by
\begin{equation}\label{eq:CarlemanA}
\begin{aligned}
    B_{j,j+1} &=
     E_{2,s}\otimes I^{\otimes (j-1)}
    + I\otimes  E_{2,s}\otimes I^{\otimes (j-2)}
    + \cdots
    + I^{\otimes (j-1)}\otimes E_{2,s},\\
    B_{j,j} &=
     E_{1,s}\otimes I^{\otimes (j-1)}
    + I\otimes  E_{1,s}\otimes I^{\otimes (j-2)}
    + \cdots
    + I^{\otimes (j-1)}\otimes E_{1,s},\\
    B_{j,j-1} &=
     E_{0,s}\otimes I^{\otimes (j-1)}
    + I\otimes  E_{0,s}\otimes I^{\otimes (j-2)}
    + \cdots
    + I^{\otimes (j-1)}\otimes  E_{0,s}.
\end{aligned}
\end{equation}
The truncation order $N$ is chosen so that the Carleman truncation error is below the accuracy $\epsilon$. This choice also contributes to the query complexity of the quantum solvers for the linearized differential equation. The detailed discussion on error analysis can be found in~\cref{sec:convergenceMPS_stable}.

\item \textbf{Quantum solver for linear differential equations.}
We solve the lifted linear differential equation in~\eqref{eq:carleman-trunc-LM} with a QLSA-based method for linear differential equations based on truncated Taylor series~\cite{An_2026,berry2014high}.
We briefly summarize the construction and refer to~\cite{An_2026} for further details.

By Duhamel's principle, the solution of the linear ODEs~\cref{eq:carleman-trunc-LM} is
\begin{equation}
    \hat{z}(t) = e^{tB_N} \hat{z}(0) + \int_0^t {\mathrm{d}\tau} e^{(t-\tau)B_N} d_N.
\end{equation}
Let $M$ be the number of time steps and set $h = T/M$, with grid points $t_m = mh$.
 Approximating the matrix exponential by a truncated Taylor series of order $J$, and
 denoting by $\hat{z}^{(m)}$ the approximation to $\hat{z}(t_m)$, we obtain the one step iteration formula
\begin{equation}
    \hat{z}^{(m+1)} = \sum_{j=0}^J \frac{1}{j!}(B_N h)^j \hat{z}^{(m)} + \sum_{j=0}^{J-1} \frac{1}{(j+1)!}(B_N)^{j}h^{j+1} d_N = R \hat{z}^{(m)} + p .
\end{equation}
The linear system is constructed through concatenating the approximate solutions at different time steps. By appending an extra $M_p - 1$ rows to the linear system, the success probability of measuring the final state can be amplified to a constant. The resulting linear system takes the form
\begin{equation}
    A_{M,M_p-1} Y = b_{M,M_p-1},
\end{equation}
where
\begin{equation}\label{eq:AMMp}
A_{M,M_p-1}=
\left(
\begin{array}{ccccccccc}
I \\
-R & I \\
& & \ddots & \ddots \\
& & & -R & I \\
& & & & -I & I\\
& & & & & & \ddots & \ddots\\
& & & & & & & -I & I\\
\end{array}
\right),
\quad
b_{M,M_p - 1}=
\left(
\begin{array}{c}
\hat{z}(0) \\
p \\
\vdots \\
p\\
0\\
\vdots\\
0
\end{array}
\right), \quad Y
=
\begin{pmatrix}
\hat{z}^{(0)}\\
\hat{z}^{(1)}\\
\vdots\\
\hat{z}^{(M)}\\
\hat{z}^{(M)}\\
\vdots\\
\hat{z}^{(M)}\\
\end{pmatrix}.
\end{equation}
The corresponding exact solution is
\[
Y_{\mathrm{exact}}
=
\bigl[
\hat{z}(0);\,
\hat{z}(h);\,
\hat{z}(2h);\,
\cdots;\,
\hat{z}(Mh);\hat{z}(Mh);\hat{z}(Mh);\cdots;\hat{z}(Mh)
\bigr].
\]
The QLSA  {outputs} a quantum state approximating $Y_{\mathrm{exact}}/\|Y_{\mathrm{exact}}\|$. Through the amplitude amplification incorporated in the linear differential equation solver, we obtain a quantum state $\ket{\tilde{z}}$ approximating the normalized final time lifted state $\ket{\hat{z}(T)}$. To approximate the solution in $\mathbb{R}^n$, we extract the first block of $\ket{\tilde{z}}$.
After an additional round of amplitude amplification, the success probability of this block is boosted to a constant. The state, denoted by $\ket{\tilde{z}_1}$ or $\ket{\hat{v}_T}$, approximates $\ket{v(T)}$ up to precision $\epsilon$,
 \begin{equation}
     \|\ket{\hat{v}_T} - \ket{v(T)}\|<\epsilon.
 \end{equation}
To implement the linear differential equation solver, we require a block encoding of $B_N$ and the state preparation oracles for $\ket{d_N}$ and $\ket{\hat{z}(0)}$.
These oracles are built from the block encodings of $F_{i}$ and $Q$, along with the state preparation oracle for the initial state $x(0)$ and pivot $s$, see details in~\cref{sec:QLSAbasedalgorithm}.

Finally, we recover the original solution by applying the inverse Lyapunov transform and then undo the pivot shift to the state $\ket{\hat{v}_T}$,
\begin{equation}
    \ket{u_T} = \ket{Q^{-1}\ket{\hat{v}_T}}, \ket{x_T} = \ket{ {u_T}+s}.
\end{equation}
Using the shift oracle in~\cref{lem:query0} and block encodings of $Q$ and $Q^{-1}$, the complete algorithm starts from the state preparation oracle $O_{x_0}$ for $\ket{x(0)}$ and output {s} a quantum state approximating $\ket{x(T)}$.
\end{enumerate}

\section{Complexity}
\label{sec:complexity}
In this section, we consider the pivot-shifted Carleman linearization algorithm introduced in~\cref{sec:algorithm} and prove the query complexity result stated in the main result. As discussed in~\cref{sec:preliminaries}, we assume access to the block encoding oracles for $F_{2}$, $F_{1}$, $F_{0}$ and $Q$, together with state preparation oracles for the initial state $\ket{x(0)}$ and the pivot $\ket{s}$. The complexity of the algorithm is measured in terms of the number of queries to these oracles.

 {Throughout the section, we express the complexity bounds in terms of the following parameter,
\begin{equation}\label{eq:gv_alphae}
    g_v = \frac{\max_{t\in[0,T]}\|v(t)\|}{\|v(T)\|}, \quad \alpha_E = \alpha_{F_{0}}\alpha_Q + \alpha_{F_{1}}\left(\|s\|\alpha_Q +\kappa_Q\right)+\alpha_{F_{2}}\left(\|s\|^2\alpha_Q + 2\|s\|\kappa_Q+\frac{\kappa_Q^2}{\alpha_Q}\right)\,,
\end{equation}
where $\kappa_Q$ denotes the condition number of the matrix $Q:=\sqrt{P}/\gamma$. }

We establish sublinear time dependence for systems that become stable after pivot shift.
For general systems that  {may} remain unstable even after pivot shift, we prove the short time convergence.
We begin with the stable case where we choose the Lyapunov transform properly so that $\mu(B_N)<0$, defined in~\cref{eq:lognorm}.

\begin{thm}[Main result 1: stable system]\label{thm:main-stable}
Suppose Assumptions~\ref{assumption:stability} and~\ref{assumption:stability2} hold and $ {s\neq x(0)}$, we choose $Q = \sqrt{P}/\gamma$ with rescaling parameter $\gamma \in \left(\zeta_{-}^P, r_+^P\right)$, where
\begin{equation}\label{eq:gamma_bd}
    r_\pm^P=\frac{-\mu_P( F_{1,s})\pm \sqrt{\mu_P( F_{1,s})^2 -4\| F _{2,s}\|_P\| F_{0,s} \|_P}}{2\| F_{2,s}\|_P},\quad \zeta_{-}^P= \frac{-4\mu_P(F_{1,s})-\sqrt{16\mu_P(F_{1,s})^2-60\|F_{0,s}\|_P\|F_{2,s}\|_P}}{6\|F_{2,s}\|_P}\,.
\end{equation}
Here, $F_{i,s}$ are defined in~\cref{eq:Fis} with $\mu_P(\cdot)$ and $\|\cdot\|_P$ defined in~\cref{sec:preliminaries}.

Denote the Carleman truncation order $N$. Suppose the initial value $\|v(0)\|<1$, then for arbitrary $\epsilon>0$, there exists an efficient quantum algorithm  {that outputs} an $\epsilon$-approximation of the normalized quantum state $\ket{x(T)}$ using
\begin{itemize}
        \item $\widetilde{O}\left(\frac{\sqrt{\|x(T)-s\|^2+ \|s\|^2 } }{\|x(T)\|}\frac{\kappa_Q\sqrt{T}g_v}{\sqrt{|C_E|}}\alpha_E\polylog\left(\frac{ {8\|u(T)\|}\kappa_Q}{ {\min\{\|u(T)\|, \|x(T)\|\}}\epsilon}\right)\right)$ queries to  $ {U_{F_{i}}}$,
        \item $\widetilde{O}\left(\frac{\sqrt{\|x(T)-s\|^2+ \|s\|^2 } }{\|x(T)\|}\frac{\sqrt{\|x(0)\|^2 + \|s\|^2}}{\|x(0)-s\|}\frac{\kappa^2_Q\sqrt{T}g_v}{\sqrt{|C_E|}}\alpha_E\polylog\left(\frac{ {8\|u(T)\|}\kappa_Q}{ {\min\{\|u(T)\|, \|x(T)\|\}}\epsilon}\right)\right)$ queries to the state preparation oracle of $O_{x_0}$ and $O_s$,
        \item $\widetilde{O}\left(\frac{\sqrt{\|x(T)-s\|^2+ \|s\|^2 } }{\|x(T)\|}\frac{\kappa_Q^{ {2}}\sqrt{T}g_v}{\sqrt{|C_E|}}\alpha_E\polylog\left(\frac{ {8\|u(T)\|}\kappa_Q}{ {\min\{\|u(T)\|, \|x(T)\|\}}\epsilon}\right)\right)$ queries to  $ {U_Q}$,
    \end{itemize}
    Here,  {the coefficients $g_v$ and $\alpha_E$ are defined in~\cref{eq:gv_alphae}, and}
    the constants are defined with weighted norm and logarithmic norm on $O = P/\gamma^2$,
    \begin{equation}\label{eqn:C_E}
    C_E:=\max\{ 4\mu_O(F_{1,s}) + 3\| F_{2,s}\|_O + 5\| F_{0,s}\|_O,
        \mu_O( F_{1,s}) +\| F_{2,s}\|_O + \| F_{0,s}\|_O\}<0\,.
    \end{equation}
\end{thm}

The condition in~\cref{thm:main-stable} requires prior knowledge of both the stable region and a suitable transformation matrix $Q$. Under these assumptions, the shifted system provably satisfies the conditions needed both for the convergence of the Carleman linearization and for the efficient application of a QLSA-based solver for the resulting dissipative linear differential equation.

The next case is when the underlying system is unstable or no prior knowledge about the stable region is available. In this case, we use the classical ``shift-to-origin'' trick to stabilize the Carleman linearization~\cite{WeberMathis2016CarlemanValidity,KowalskiSteeb1991Carleman} by choosing the pivot to be the initial state, $s = x(0)$, which guarantees good local behavior and thus supports convergence of the Carleman approximation over short time scales. However, since the shifted system is not guaranteed to be stable, we can only expect short-term convergence of the Carleman linearization, and the complexity of the resulting QLSA-based linear differential equation solver will have a worse dependence on the final time $T$. The complexity in this case is summarized in the following theorem.

\begin{thm}[Main result 2: unstable system]\label{thm:main-unstable}
Choose $s=x(0)$ and $Q = I$,  
then there exists $t^*>0$ such that: 1. $\|v(t)\|<1$ for all $0\leq t\leq t^*$; 2. For any $T<t^*$, there exists an efficient quantum algorithm to solve the quadratic ODE and output a quantum state $\epsilon$-close to the final state $\ket{x(T)}$ with
\[
        \widetilde{\mathcal{O}}\left(\frac{\sqrt{\|x(T)-x(0)\|^2+ \|x(0)\|^2 } }{\|x(T)\|}\frac{\left(
                    \alpha_{F_0}+\alpha_{F_1}+\alpha_{F_2}
                \right)(1+\|x(0)\|)^2\left\|F_{2}\right\|T^2g_v}{ {\min\{\|x(T)-x(0)\|, \|x(T)\|\}}\epsilon}\right) 
\]
queries to $ {U_{F_i}}$ and $O_{x_0}$.
\end{thm}
The performance guarantee holds up to a time point $t^*$, which may decrease when the shifted coefficient scale {s}, namely, when $\sum_i\|F_{i,s}\|$ becomes large.

We put the proof in the following part of this section. The complexity analysis consists of four main steps:
\begin{itemize}
  \item \emph{Step 1: Preparation of oracles for the shifted system.}

  This corresponds to the first two steps of the algorithm, namely the pivot shift and the Lyapunov transform. Specifically, we need to construct the block encoding oracles for $ {E}_{2,s}$, $ {E}_{1,s}$, and $ {E}_{0,s}$, as well as the state preparation oracle for $\ket{ {v}(0)}$. These are constructed using the block encoding for $F_{2}$, $F_{1}$, $F_{0}$ and $Q$, along with the state preparation oracle for $\ket{s}$ and $\ket{x(0)}$. The complexity of this step is analyzed in~\cref{sec:preliminaries}~\cref{lem:query0}.

  \item \emph{Step 2: Convergence of the Carleman linearization.}

  This corresponds to the third step of the algorithm, where we analyze the error convergence of the Carleman linearization for the shifted system. We derive explicit bounds on the truncation order $N$ required to achieve a target accuracy $\epsilon$ for approximating the solution at time $T$. The analysis is presented in~\cref{sec:convergenceMPS_stable}. For stable systems, this result can be stated in a consistent form with prior knowledge, which is given in~\cref{prop:stable_carleman_convergence}. For unstable systems, we only expect short time convergence, and the convergence result is stated in~\cref{prop:unstable_carleman_convergence}.

  \item \emph{Step 3: Complexity of solving the lifted linear differential equations.}

  This corresponds to the fourth step of the algorithm, where we analyze the complexity of applying the QLSA-based linear differential equation solver to the lifted system obtained from the Carleman linearization. In this step, we will mainly rely on the algorithms in~\cite{An_2026,kroviImprovedQuantumAlgorithms2023} to obtain a better dependence on the final time $T$. The complexity result is summarized in~\cref{sec:QLSAbasedalgorithm}~\cref{prop:stable_linear_solver} for stable systems and in~\cref{prop:unstable_linear_solver} for unstable systems.

  \item \emph{Step 4: Putting everything together and shifting the state back.}

  Finally, we combine the results from the previous steps to derive the overall complexity of the pivot-shifted Carleman linearization algorithm for both stable and unstable systems as stated in~\cref{thm:main-stable}, and~\cref{thm:main-unstable}. Here, we also need to analyze the complexity of shifting the state back from $\ket{v(T)}$ to $\ket{x(T)}$. 
  This is done in~\cref{sec:proofMain}.
\end{itemize}

\subsection{Convergence of the Carleman linearization}
\label{sec:convergenceMPS_stable}
In this section, we analyze the convergence of the Carleman linearization for the shifted system in both the stable and unstable cases. This completes the second step of the complexity analysis outlined above. We start with the error caused by the Carleman linearization with truncation order $N$.  We denote the exact lifted solution from the Carleman linearization with Lyapunov transform in~\eqref{eq:carleman-trunc-LM} algorithms as  $z(t) = (z_k(t))_{k\ge 1}$ where $z_k(t) = v(t)^{\otimes k}$. Following the algorithm stated in~\cref{sec:algorithm}, we denote the corresponding error vectors as  $\xi = \xi(t)\coloneqq z(t) - \hat z(t)$, satisfying the equation
\begin{equation}\label{eq:CarlemanLinearized}
    \partial_t \xi = B_N \xi + \hat d_N,\quad  {\xi(0) = 0}\,,
\end{equation}
with the matrix $B_N$ in~\cref{eq:carleman-trunc-LM}, and
\begin{equation}\label{eq:hatd}
\hat d_N
=
\bigl[
0;\,
0;\,
\cdots;\,
0;B_{N,N+1}z_{N+1}(t)
\bigr]\,.
\end{equation}
By analyzing the bound on $\|\xi_1(t)\|:=\|z_1(t)-\hat{z}_1(t)\|$, we can obtain the error bound for approximating $v(t)$.

First, for the stable solution, as will be shown in~\cref{lem:solutionBdP}, under Assumptions of~\cref{thm:main-stable}, the  {norm of the} solution $v(t)$ of the shifted system can be bounded by a constant that is strictly less than 1 for all $t\in[0,T]$. This bound on the solution allows us to obtain the exponential decay of the Carleman error with respect to the truncation order $N$, which in turn leads to a logarithmic scaling of $N$ with $T/\epsilon$ to achieve an $\epsilon$-accurate solution. The precise statement is given in the following proposition.

\begin{prop}[Stable system]\label{prop:stable_carleman_convergence}
Under Assumptions~\ref{assumption:stability} and~\ref{assumption:stability2}, we choose $Q = \sqrt{P}/\gamma$ with rescaling parameter $\gamma \in \left(\zeta_{-}^P, r_+^P\right)$ and denote the Carleman truncation order $N$. Suppose the initial value $\|v(0)\|<1$. Then,
\[
\max_{t\in[0,T]}\|v(t)\|\leq \max\left\{\|v(0)\|,\frac{r^P_-}{\gamma}\right\}<1\,.
\]
Furthermore, the  {norm of the} Carleman embedding error vector can be bounded:
\begin{equation}\label{eq:xi_bound2}
\| \xi(t)\|
 \le t N\gamma\|F_{2,s}\|_{P}\max_{t\in[0,T]}\|v(t)\|^{N+1}\,,
\end{equation}
The relative error can be bounded by $\epsilon$, meaning $\|\xi_1(T)\|\leq \|\xi(T)\|\leq \epsilon \|v(T)\|$, by
choosing the truncation order
\[
N=\Theta\left(\frac{\log\left(T\gamma\|F_{2,s}\|_{P}/(\|v(T)\|\epsilon)\right)}{\log(1/\max_{t\in[0,T]}\|v(t)\|)}\right)\,.
\]
\end{prop}

We provide the proof of~\cref{prop:stable_carleman_convergence} in~\cref{sec:solutionBdd}.~\cref{prop:stable_carleman_convergence} implies that, as long as $x(0)-s$ is sufficiently small and the system is stable, the truncation order $N$ only needs to scale logarithmically with $t/\epsilon$ to achieve an $\epsilon$-accurate solution.

Without Assumptions~\ref{assumption:stability} and~\ref{assumption:stability2} and correct choice of $Q$, the system might not be a stable system, while we can still provide the bound for Carleman error vector for short time.
When $\|v(0)\|$ is sufficiently small, combining the Carleman error with the solution bound, we can also obtain a logarithmic scaling of the truncation order $N$ with $t/\epsilon$ to achieve an $\epsilon$-accurate solution for the unstable system for short time, as shown in the following proposition.
\begin{prop}[Unstable system]\label{prop:unstable_carleman_convergence}
Choose $Q=I$ and pivot $s=x(0)$ and  {$t^*  {=} \frac{1}{ {e^2}\left(\|F_{1,s}\| + \|F_{0,s}\| + \|F_{2,s}\|\right)}\,.$}
 {For time $0\leq t {<} t^*$,} the solution satisfies  $\|v(t)\|<1$, and the  {norm of the} Carleman embedding error vector is  {bounded} by
\begin{equation}
    \|\xi(t)\|\leq tN^{3/2}\|F_{2,s}\| \left(t\left(\|F_{2,s}\| +\|F_{1,s}\| + \|F_{0,s}\|\right)e\right)^{N+1}\leq tN^{3/2}\|F_{2,s}\|\frac{1}{e^{N+1}}.
\end{equation}
The relative error can be bounded by $\epsilon$ so that $\|\xi(t)\|\leq \epsilon \|v(t)\|$ by
choosing the truncation order
\[
N \geq \Theta\left(\frac{\log(t\|F_{2,s}\|/(\|v(t)\|\epsilon))}{\log(1/\left(t^*\left(\|F_{2,s}\| + \|F_{1,s}\| + \|F_{0,s}\|\right)e\right))}\right).
\]
\end{prop}
The proof of~\cref{prop:unstable_carleman_convergence} is similar to that of~\cref{prop:stable_carleman_convergence}. We put the proof of~\cref{prop:unstable_carleman_convergence} in Appendix~\ref{appd:unstable}.

\subsubsection{Proof of \texorpdfstring{\cref{prop:stable_carleman_convergence}}{Lg}}
\label{sec:solutionBdd}
We first establish a bound on the  {norm of the} solution $v(t)$ of the shifted system, which is crucial for analyzing the convergence of the Carleman linearization.
\begin{lem}\label{lem:solutionBdP}
    Under Assumptions~\ref{assumption:stability} and~\ref{assumption:stability2}. Consider $v(t)$ in~\cref{eqn:v_t_equation}. Define
    \begin{equation}\label{eq:rpm_P}
   \begin{aligned}
   r_\pm^P & =\frac{-\mu_P( F_{1,s})\pm \sqrt{\mu_P( F_{1,s})^2 -4\| F _{2,s}\|_P\| F_{0,s} \|_P}}{2\| F_{2,s}\|_P}\,.\\
   \end{aligned}
   \end{equation}
   Then,
\begin{equation}\label{eq:Dv}
    \max_{t\in[0,T]}\|v(t)\|  {\leq} \begin{cases}
        \|v(0)\| & \text{if }\|v(0)\|\in[\frac{1}{\gamma}r_-^P, \frac{1}{\gamma}r_+^P),\\
         \frac{1}{\gamma} r_-^P&  \text{if }\|v(0)\|\in[0, \frac{1}{\gamma}r_-^P).\\
    \end{cases}
    \end{equation}
In particular, if $\gamma \in \left(\zeta_{-}^P,r_+^P\right)\subset\left(r_{-}^P, r_+^P\right)$ and $\|v(0)\|<1$, then $\max_{t\in[0,T]}\|v(t)\|\leq \max\left\{\|v(0)\|,\frac{r^P_-}{\gamma}\right\}<1$.
\end{lem}

This bound in~\cref{lem:solutionBdP} is slightly more general than the results in~\cite{jennings_quantum_2025,liu_efficient_2021}, as it does not require a positive lower bound on $\|u(0)\|$ (i.e., an assumption of the form $\delta < \|u(0)\|$ for some $\delta > 0$). For example, when $u(0) = 0 < r_-^P$, the solution is not controlled by its initial value, yet we can still obtain the uniform upper bound $\|u(t)\|_P < r_-^P$. The absence of a lower bound on $\|u(0)\|$ is consistent with the convergence radius statement in~\cite{novikau2025globalizing}, while yielding a more precise mathematical formulation.

Under these conditions, the norm of the transformed solution remains uniformly bounded by $1$ on $[0,T]$. This uniform bound enables us to establish the convergence of the Carleman linearization and prove~\cref{prop:stable_carleman_convergence} as follows.

\begin{proof}[Proof of~\cref{prop:stable_carleman_convergence}]
The Lyapunov transform matrix $Q$ works as the preconditioner matrix to the Carleman linearization matrix $B_N$.
 {The proof follows a similar technique to that presented in Appendix B.2 of \cite{jennings_quantum_2025}. For completeness, we adapt the analysis to the present setting, where boundedness of the solution norm is used in place of the decay estimate assumed in the previous analysis, and provide the details below.}

We recall the definition of the error vector $\xi(t)$ in~\cref{eq:CarlemanLinearized} and~\cref{eq:hatd}:
\[
    \partial_t \xi = B_N \xi + \hat d_N\,,
\]
with
\[
\hat d_N
=
\bigl[
0;\,
0;\,
\cdots;\,
0;B_{N,N+1}v^{\otimes (N+1)}(t)
\bigr]\,.
\]
where the block entries $B_{\cdot,\cdot}$ are defined in~\cref{eq:CarlemanA}.

Notice that the error term $\|\xi(t)\|$ follows
\begin{equation}
    \partial_t \|\xi\|^2 = \xi^\dag (B_N + B_N^\dag)\xi + \xi^\dag  \hat{d}_N + \hat{d}^\dag_N \xi\,.
\end{equation}
The first term can be estimated as
\begin{equation}\label{eq:bd_BN}
    \begin{aligned}
        &\xi^\dagger( B_N+B_N^\dagger)\xi  = \sum_j \xi_j^\dag (B_{j,j} + B_{j,j}^\dag)\xi_j + \sum_{j} \xi_{j+1}^\dag (B_{j+1, j} + B_{j, j+1}^\dag)\xi_j + \sum_j \xi_j^\dag(B_{j+1,j}^\dag + B_{j,j+1})\xi_{j+1} \\
        &  \leq \sum_j 2\mu(B_{j,j})\|\xi_j\|^2 + \sum_j2\|B_{j+1,j}\|\|\xi_j\|\|\xi_{j+1}\|+\sum_j2\|B_{j,j+1}\|\|\xi_j\|\|\xi_{j+1}\|\\
        &\leq \sum_j 2j\mu(E_{1,s})\|\xi_j\|^2 + 2j\|E_{0,s}\|\|\xi_j \|\| \xi_{j-1}\| + 2(j-1)\|E_{2,s}\| \|\xi_j \|\|\xi_{j-1}\|\\
        & \le \xi_G^\dag (G+ G^\dag)\xi_G\leq 2\mu(G)\,.
    \end{aligned}
\end{equation}
In the second inequality, we use the definition of $B_{j,j}$, $B_{j,j+1}$, and $B_{j+1,j}$ in~\cref{eq:CarlemanA}. In the last inequality, we define the vector $\xi_G = (\|\xi_i\|)_i$ and
the tridiagonal matrix $G\in\mathbb{C}^{N\times N}$:
\begin{equation}\label{eq:Gij}
    G_{j,j} = j\mu(E_{1,s}), G_{j-1,j} = j\|E_{0,s}\|, G_{j+1,j} = j\|E_{2,s}\|.
\end{equation}
We denote the upper bound of the eigenvalues of $G+ G^\dag $ as $C_E$.
By Gershgorin's circle theorem,
\begin{equation}
    C_E \geq \begin{cases}
        2G_{j,j} + (G_{j,j+1}+G_{j+1,j}) + (G_{j,j-1} + G_{j-1,j})& \text{for }j = 2, \cdots, N-1\\
        2G_{1,1} + (G_{1,2}+G_{2,1}) & \text{for } j = 1\\
        2G_{N,N} + (G_{N,N-1} + G_{N-1, N}) & \text{for }j= N\\
    \end{cases}
\end{equation}
Substituting the element of $G$ in~\cref{eq:Gij}, we can show that
the largest eigenvalue of $G+ G^\dag$ is strictly negative provided the following conditions are satisfied:
\begin{equation}
    0\geq \begin{cases}
        2j\mu(E_{1,s}) + (2j-1)\|E_{2,s}\| + (2j+1)\|E_{0,s}\|& \text{for }j = 2, \cdots, N-1\\
        2\mu(E_{1,s}) + \|E_{2,s}\|+ 2\|E_{0,s}\| & \text{for } j = 1\\
        2N\mu(E_{1,s}) + (N-1) \|E_2\| + N\|E_{0,s}\| & \text{for }j= N\\
    \end{cases}
\end{equation}
It is sufficient to have the following conditions to ensure the above three conditions are satisfied:
\begin{equation}\label{eq:cond_muGnegative}
    \begin{cases}
        \mu(E_{1,s}) + \|E_{2,s}\| + \|E_{0,s}\| &<0\\
        4\mu(E_{1,s}) + 3\|E_{2,s}\| + 5\|E_{0,s}\| &<0\\
        2\mu(E_{1,s}) + \|E_{2,s}\| + 2\|E_{0,s}\| &<0.\\
    \end{cases}
\end{equation}
The third inequality can be derived by the first two, thus
\begin{equation}\label{eq:C0}
   C_E:=\max\{ 4\mu( E_{1,s}) + 3\| E_{2,s}\| + 5\| E_{0,s}\|,
        \mu( E_{1,s}) +\| E_{2,s}\| + \| E_{0,s}\|\}<0\ \Rightarrow\ \mu(G)\leq\frac{1}{2} C_E<0\,.
\end{equation}
We will show later that this $C_E$ in~\eqref{eq:C0} matches the $C_E$ defined in~\cref{thm:main-stable}~\eqref{eqn:C_E} when $Q = \sqrt{P}/\gamma$, see~\cref{eq:E2F} for detail.

Now, let's assume~\eqref{eq:C0} is satisfied (we will show that it can be satisfied by choosing $Q$ properly later). Since \eqref{eq:bd_BN} is established for any $\xi$, we have $\mu(B_N)\leq \mu(G)$.
The forcing term $\hat{d}_N$ can be upper bounded by Cauchy--Schwarz inequality and $\|B_{N,N+1}\| = \sum_{m = 1}^{N} \|I^{\otimes (m-1)}\otimes E_{2,s}\otimes I^{N-m}\| \leq \sum_{m=1}^N \|E_{2,s}\| = N\|E_{2,s}\|$ as
\begin{equation}
\begin{aligned}
     &\xi^\dag  \hat{d}_N + \hat d^\dag_N \xi = 2\mathrm{Re}(\xi_N B_{N,N+1} v^{\otimes (N+1)}(t))\leq 2N\|\xi\| \|E_{2,s}\|\|v(t)\|^{N+1}
\end{aligned}
\end{equation}
Consequently, we have $\frac{d}{dt}\| \xi\|^2 \leq 2\mu(G)\|\xi\|^2 + 2N\|\xi\| \|E_{2,s}\|\|v\|^{N+1}$, it reveals that for $\|\xi\|>0$, we have
\begin{equation}
    \frac{d}{dt}\| \xi\| \leq \mu(G)\|\xi\| + N \|E_{2,s}\|\|v(t)\|^{N+1},\quad \|\xi(0)\|=0\,.
\end{equation}
By Gr\"onwall's inequality, we obtain
\begin{equation}\label{eq:xibound}
    \|\xi(t)\|\leq tN \|E_{2,s}\|\max_{t\in[0,T]}\|v(t)\|^{N+1}.
\end{equation}

Finally, when $Q = \sqrt{P}/\gamma$, we note that
\[
\mu(E_{1,s}) = \mu_P(F_{1,s}),\quad
    \|E_{2,s}\|= \gamma\|F_{2,s}\|_{P},\quad \|E_{0,s}\| = \frac{1}{\gamma}\|F_{0,s}\|_P\,.
\]
The condition in~\cref{eq:C0} is then satisfied when $\gamma \in \left(\zeta_{-}^P, r_+^P\right)$. This concludes the proof of~\cref{prop:stable_carleman_convergence}.
\end{proof}

In the remaining part of this section, we will provide the proof of~\cref{lem:solutionBdP}. We first provide a bound for general Lyapunov transformation and then show that the special case reveals the results in~\cref{lem:solutionBdP}.
\begin{lem}[Stable system: long time bound of $\|v(t)\|$]\label{lm:normboundMPS} Define $O = Q^\dag Q$, where $Q$ is the Lyapunov transform matrix defined in~\cref{eq:LyapunovTransform}.
  Suppose $\mu_O(F_{1,s})<0$ and $ \mu_O(F_{1,s})^2 -4\|F_{2,s}\|_O\|F_{0,s} \|_O>0$. Consider $u(t)$ in~\cref{eq:MPSquadratic} and $v(t)$ in~\cref{eqn:v_t_equation}. Define
    \begin{equation}
   \begin{aligned}
   r_\pm^O & =\frac{-\mu_O( F_{1,s})\pm \sqrt{\mu_O( F_{1,s})^2 -4\| F _{2,s}\|_O\| F_{0,s} \|_O}}{2\| F_{2,s}\|_O}\,.\\
   \end{aligned}
   \end{equation}
   Then,
   \begin{itemize}
       \item If $\|u(0)\|_O=\|v(0)\|\in \left[r_-^O, r_+^O\right)$, then $\|u(t)\|_O =\|v(t)\|\leq \|u(0)\|_O< r_+^O$.
       \item If $\|u(0)\|_O =\|v(0)\|\in \left[0,r_-^O\right)$, then $\|u(t)\|_O =\|v(t)\|< r_-^O$.
   \end{itemize}
\end{lem}
The proof of the solution bound follows the comparison theorem and the bound for the scalar Riccati equation.
\begin{proof}[Proof of \cref{lm:normboundMPS}]
We first estimate the time derivative of $\|v\|^2 = v^\dag v $,
    \begin{equation}
    \begin{aligned}
        \partial_t \|v\|^2 &= v^\dag E_{2,s} (v\otimes v) + (v^\dag \otimes v^\dag) E_{2,s}^\dag v + v^\dag (E_{1,s} + E_{1,s}^\dag)v + v^\dag E_{0,s} + E_{0,s}^\dag v \\
        &\le 2\|E_{2,s}\|\|v\|^3 + 2\mu(E_{1,s})\|v\|^2 + 2\|E_{0,s} \|\|v\|\,.
    \end{aligned}
    \end{equation}
    If $\|v\|\neq 0 $, this implies
    \begin{equation}
        \partial_t \|v(t)\| \leq\|E_{2,s}\|\|v\|^2 + \mu(E_{1,s})\|v\| + \|E_{0,s} \| \,.
    \end{equation}
    Consider the differential equation for $ {y}(t)$,
    \begin{equation}
        \partial_t  {y}  =  \|E_{2,s}\| {y}^2 + \mu(E_{1,s}) {y} + \|E_{0,s} \|, \quad  {y(0) = \|v(0)\|,}
    \end{equation}
     {and $\|v(t)\|\leq y(t)$ from the comparison theorem.} 
    When its discriminant is positive,
    \begin{equation}
        \Delta =  \mu(E_{1,s})^2 -4\|E_{2,s}\|\|E_{0,s} \|>0\,,
    \end{equation}
    the quadratic polynomial admits two distinct real roots
    \begin{equation}\label{eq:rpm}
    r_\pm \coloneqq \frac{-\mu(E_{1,s})\pm \sqrt{\mu(E_{1,s})^2 -4\|E_{2,s}\|\|E_{0,s} \|}}{2\|E_{2,s}\|}\,.
\end{equation}
From the relationship between $E_{i,s}$ and $F_{i,s}$ in~\cref{eq:EfromF} and the properties of the weighted inner product introduced in~\cref{sec:preliminaries}, we have
\begin{equation}
    \mu(E_{1,s}) = \mu_O(F_{1,s}), \|E_{0,s}\| = \|F_{0,s}\|_O, \|E_{2,s}\| = \|F_{2,s}\|_O, \|v(t)\|=\|u(t)\|_O\,.
\end{equation}
Substituting them into the expression of $r_\pm$ in~\cref{eq:rpm}, we can rewrite those two roots as
\begin{equation}
    r_\pm = \frac{-\mu_O( F_{1,s})\pm \sqrt{\mu_O( F_{1,s})^2 -4\| F _{2,s}\|_O\| F_{0,s} \|_O}}{2\| F_{2,s}\|_O} = r_\pm^O\,.
\end{equation}
By the comparison theorem, we obtain the bound on $\|v(t)\|$ and its dependence on $\|v(0)\|$, i.e.
\begin{itemize}
    \item If $\|v(0)\| \in \left[r_-^O, r_+^O\right)$, then $\|v(t)\|\leq \|v(0)\|< r_+^O$.
    \item If $\|v(0)\|\in \left[0,r_-^O\right)$, then $\|v(t)\|< r_-^O$.
\end{itemize}
With $\|v(0)\| = \|u(0)\|_O, \|v(t)\| = \|u(t)\|_O$, this completes the proof.
\end{proof}

Now, we are ready to prove~\cref{lem:solutionBdP} by applying~\cref{lm:normboundMPS} with a specific choice of $Q$:
\begin{proof}[Proof of~\cref{lem:solutionBdP}]
The conditions in~\cref{lm:normboundMPS} are satisfied by choosing $Q = \sqrt{P}/\gamma$ under Assumptions~\ref{assumption:stability} and~\ref{assumption:stability2}. From Assumptions~\ref{assumption:stability}, we have
\[
    P F_{1,s}  + F_{1,s}^\dag P \prec 0.
\]

From the property of log norm in~\cref{sec:preliminaries}, the weighted log norm is negative as
\begin{equation}
    \mu_O(F_{1,s}) = \frac{1}{2}\max_i\lambda_i\left( \sqrt{P}^{-1}\left( PF_{1,s}+ F_{1,s}^\dag P\right)\sqrt{P}^{-1}\right)<0.
\end{equation}
because $P$ is Hermitian.  Under Assumptions~\ref{assumption:stability2}, we have
\[
   \mu_P(F_{1,s})^2 -4\|F_{2,s}\|_P\|F_{0,s} \|_P>0.
\]
Furthermore, with $Q = \sqrt{P}/\gamma$, we have
\begin{equation}\label{eq:E2F}
\mu_O(F_{1,s}) = \mu_P(F_{1,s}),\quad
    \|F_{2,s}\|_O = \gamma\|F_{2,s}\|_{P},\quad \|F_{0,s}\|_O = \frac{1}{\gamma}\|F_{0,s}\|_P,\quad \|v(t)\| = \|u(t)\| {_O}\,.
\end{equation}
The first three equations in~\cref{eq:E2F} imply that $\mu_O(F_{1,s})^2 -4\|F_{2,s}\|_O\|F_{0,s} \|_O>0$ and $r_\pm^O = \frac{1}{\gamma} r^P_\pm$. This concludes the proof of~\cref{lem:solutionBdP}.
\end{proof}

\subsection{Quantum solvers for linear differential equations}
\label{sec:QLSAbasedalgorithm}
In this section, we analyze the query complexity of the quantum algorithm applied to the linear differential equation in~\cref{eq:carleman-trunc-LM}. This completes the third step of the complexity analysis. According to~\eqref{eq:bd_BN},\eqref{eq:C0},\eqref{eq:E2F} in the proof of~\cref{prop:stable_carleman_convergence}, we have
\[
    B_N + B_N^\dag\leq \frac{1}{2}C_E:=\frac{1}{2}\max\{ 4\mu_O(F_{1,s}) + 3\| F_{2,s}\|_O + 5\| F_{0,s}\|_O,
        \mu_O( F_{1,s}) +\| F_{2,s}\|_O + \| F_{0,s}\|_O\}<0
\]
under the assumption of~\cref{thm:main-stable}. By applying the quantum algorithm~\cite{An_2026} for dissipative differential equations, we obtain the following result for solving the Carleman linearized system.

\begin{prop}[Stable system]\label{prop:stable_linear_solver}
    Under the assumption of~\cref{thm:main-stable}. Let $\epsilon>0$ and assume the simulation time $T\geq 1/|C_E|$. There exists a quantum algorithm that outputs an $\epsilon$-approximation of the final state $\ket{v(T)}$ using
    \begin{equation}
        \widetilde{O}\left(\frac{\sqrt{T}g_v}{\sqrt{|C_E|}}\alpha_E\polylog(1/\epsilon)\right)
    \end{equation}
    queries to the block-encoding of $F_{i,s}$, for $i=0,1,2$, and the state preparation oracle of $\ket{v(0)}$,
    \begin{equation}
        \widetilde{O}\left(\frac{ {\kappa_Q}\sqrt{T}g_v}{\sqrt{|C_E|}}\alpha_E\polylog(1/\epsilon)\right)
    \end{equation}
    queries to $ {U_Q}$. Here $\kappa_Q$ is the condition number of $Q=\sqrt{P}/\gamma$ used in~\cref{thm:main-stable}.
\end{prop}

In general, the shifted system may remain unstable even after pivot shift. Nevertheless, over a short time interval, one can still establish the following query complexity using the quantum algorithm for linear differential equations~\cite{kroviImprovedQuantumAlgorithms2023}.
\begin{prop}[Unstable system]\label{prop:unstable_linear_solver}
  Under the assumption of~\cref{thm:main-unstable},  
  we choose $t^*$ as in~\cref{prop:unstable_carleman_convergence}. For any $T {<} t^*$, there exists a quantum algorithm to solve the quadratic ODE and output a quantum state $\epsilon$-close to the final state $\ket{v(T)}$ using
    \begin{equation}
        \widetilde{\mathcal{O}}\left(g_v\frac{\left(
                    \alpha_{F_{0,s}}
                    + \alpha_{F_{1,s}}
                    + \alpha_{F_{2,s}}
                \right)\left\|F_{2,s}\right\|T^2}{ {\|v(T)\|}\epsilon}\right)
    \end{equation}
 queries to the block-encoding of $F_{i,s}$. Here, $\alpha_{F_{i,s}}$ is the block-encoding factor for $F_{i,s}$, defined in~\cref{lem:query0}.
\end{prop}
We put the proof of~\cref{prop:unstable_linear_solver} in~\cref{appd:unstable}.

Before proving~\cref{prop:stable_linear_solver}, we quantify the normalizing coefficient for $B_N$ and $d_N$ in~\cref{eq:carleman-trunc-LM}. Follows the argument in~\cite[(L.3)-(L.6)]{jennings_quantum_2025}, we have
\begin{itemize}
    \item $\alpha_{d_N} = \|E_{0,s}\|=\alpha_{F_{0,s}}\alpha_{Q}$,
    \item $\alpha_{B_N}= N(\alpha_{F_{0,s}}\alpha_Q + \alpha_{F_{1,s}}\kappa_Q+\frac{\alpha_{F_{2,s}}\kappa_Q^2}{\alpha_Q})$.
\end{itemize}
By LCU and quantum singular value transformation~\cite{gilyen2019quantum},
a single query to the block encoding oracle of $B_N$ requires one call to each block encoding of $F_{i,s}, i = 0,1,2$ and  {$\mathcal{O}(\kappa_Q\log(1/\epsilon))$} calls to the block encoding of $Q$ with precision $\epsilon$.
Moreover, one query to the state preparation oracle for $\ket{d_N}$ requires one call to $E_{0,s}$ and one call to the block  {encoding} of $Q$. From \cite[Lemma 5]{liu_efficient_2021}, the initial state $\hat{z}(0)$ can be prepared with $N$ queries to the state preparation oracle for $\ket{v(0)}$.
\begin{proof}[Proof of \cref{prop:stable_linear_solver}]
The complexity of QLSA based Carleman linearization method can be obtained by applying QLSA in Corollary 12 in~\cite{An_2026}. We first state the result here.
\begin{lem}(The time independent version for final state preparation in ~\cite[Corollary 12]{An_2026})
    Using the truncated Taylor series method for solving the ODEs $\partial_t u = A u + b$, $t\in[0,T]$ with initial condition $u(0) = u_0$. Suppose the $A$ is a matrix such that $A + A^\dag \leq -2\eta<0$. Let $\epsilon>0$ be the target error, $T\geq \eta^{-1}$ be the evolution time, $\alpha_A$ be the block-encoding factor of $A$ and $\alpha_b$ be the normalization factor of $b$. Then we can obtain an $\epsilon$-approximation of the final state $u(T)/\|u(T)\|$ using
    \begin{equation}
        \widetilde{O}\left(\frac{\max_t\|u(t)\|}{\|u(T)\|}\frac{\alpha_A \sqrt{T}}{\sqrt{\eta}}\log\left(\frac{\|u_0\| + \alpha_b}{\|u(T)\|}\frac{\alpha_A}{\eta\epsilon}\right)\log\left(\frac{1}{\epsilon}\right)\right)
    \end{equation}
    queries to the block encoding of $A$ and
    \begin{equation}
        \widetilde{O}\left(\frac{\max_t\|u(t)\|}{\|u(T)\|}\frac{\alpha_A \sqrt{T}}{\sqrt{\eta}}\log\left(\frac{1}{\epsilon}\right)\right)
    \end{equation}
    queries to the state preparation of $\ket{u_0}$ and $\ket{b}$.
\end{lem}

In our setting, we first apply the above lemma to solve the Carleman linearized equation $\partial_t \hat{z}  = B_N\hat{z} + d_N$ and then calculate the extra rounds from amplitude amplification to extract the approximation of $v(T)$. We denote the quantum state after applying the QLSA as  $\tilde{z}$ then the first block $\tilde{z}_1$, the same block when Carleman truncation order $N =1$, approximates $v(T)$.
First, we denote the precision from QLSA as $\epsilon'$.
For all $\epsilon'>0$ and $T\geq1/|C_E|$,
    \begin{equation}
        \left\|\tilde{z} - \frac{\hat{z}(T)}{\|\hat{z}(T)\|}\right\|\leq \epsilon'
    \end{equation}
    with
    \begin{equation}\label{eqn:BNquery}
        \widetilde{O}\left(\frac{\max\|\hat{z}(t)\|}{\|\hat{z}(T)\|}\frac{\alpha_{B_N} \sqrt{T}}{\sqrt{|C_E|}}\log\left(\frac{\|\hat{z}(0)\|+\alpha_{d_N}}{\|\hat{z}(T)\|}\frac{\alpha_{B_N}}{|C_E|\epsilon'}\right)\log(1/\epsilon')\right)
    \end{equation}
    queries to block-encoding of $B_N$ and
    \begin{equation}\label{eqn:initialquery}
        \widetilde{O}\left(\frac{\max_t\|\hat{z}(t)\|}{\|\hat{z}(T)\|}\frac{\alpha_{B_N} \sqrt{T}}{\sqrt{|C_E|}}\log(1/\epsilon')\right)
    \end{equation}
    queries to the state preparation oracle of $\ket{\hat{z}(0)}$ and $\ket{d_N}$. Let the Carleman linearization error be bounded by $\delta'$. \cref{prop:stable_carleman_convergence} shows that $\|\xi_j(T)\|\leq \|\xi(T)\|\leq \delta'\|v(T)\|$ can be achieved
    with truncation order $N \geq \log(T\|E_{2,s}\|/\delta'\|v(T)\|)/\log(1/\max_{t\in[0,T]}\|v(t)\|).$

    Through measuring the quantum registers, we use the first block, $\tilde{z}_1$, of size  corresponding to  {Carleman} linearization truncated at $N = 1$, to approximate $v(T)$. The success probability can be approximated as
\begin{equation}\label{eq:z1tilde}
\begin{aligned}
    \|\tilde{z}_1\|
    &\geq
    \frac{\|\hat{z}_1(T)\|}{\|\hat{z}(T)\|}
    -
    \left\|
        \frac{\hat{z}_1(T)}{\|\hat{z}(T)\|}
        - \tilde{z}_1
    \right\| \\
    &\geq
    \frac{\|\hat{z}_1(T)\|}{\|\hat{z}(T)\|}
    -
    \left\|
        \frac{\hat{z}(T)}{\|\hat{z}(T)\|}
        - \tilde{z}
    \right\|
    \geq
    \frac{\|\hat{z}_1(T)\|}{\|\hat{z}(T)\|}
    - \epsilon'.
\end{aligned}
\end{equation}
    With $\max_{t\in[0,T]}\|v(t)\| \leq 1$,
    the upper and lower bounds of $\hat{z}$ can be analyzed similar to the technique in~\cite[Theorem 8]{kroviImprovedQuantumAlgorithms2023}
    \begin{itemize}
    \item $\|\hat{z}(T)\|\geq\|\hat{z}_1(T)\|\geq \|v(T)\| - \|v(T)-\hat{z}_1(T)\|\geq\|v(T)\|(1-\delta'),$
    \item $\|\hat{z}(T)\|\leq \sqrt{N}(1+\delta') \|v(T)\|,$
    \item $\frac{\|\hat{z}_1(T)\|}{\| \hat{z}(T)\|}\geq \frac{(1-\delta')\|v(T)\|}{\sqrt{N}(1+\delta')\|v(T)\|} = \frac{(1-\delta')}{\sqrt{N}(1+\delta')}.$
\end{itemize}
Therefore, we have $\|\tilde{z}_1\| \geq \frac{(1-\delta')}{\sqrt{N}(1+\delta')} - \epsilon'$ and
\begin{equation}\label{eq:g_vdef}
    \frac{\max_t\|\hat{z}(t)\|}{\|\hat{z}(T)\|}\leq\frac{\sqrt{N}(1+\delta')\max_{t\in[0,T]}\|v(t)\|}{(1-\delta')\|v(T)\|} = \sqrt{N}\frac{1+\delta'}{1-\delta'} g_v, \quad \|\hat{z}(0)\|\leq \sqrt{N}(1+\delta')\|v(0)\|.
\end{equation}
By choosing $\epsilon' = \Theta\!\left(\frac{\epsilon}{\sqrt{N}}\right),
    \delta' = \Theta(\epsilon),$
the success probability  in~\cref{eq:z1tilde} is lower bounded by $\Theta(1/N)$, it
can be boosted to a constant using $\mathcal{O}(\sqrt{N})$ rounds of amplitude amplification.

To bound the approximation error from $\tilde{z}_1$,  we use
\begin{equation}\label{eq:z1tilde_error}
\begin{aligned}
    &\left\|
        \frac{\tilde{z}_1}{\|\tilde{z}_1\|}
        -
        \frac{v(T)}{\|v(T)\|}
    \right\|
    \leq
    \left\|
        \frac{\tilde{z}_1}{\|\tilde{z}_1\|}
        -
        \frac{\hat{z}_1(T)/\|\hat{z}(T)\|}{\|\hat{z}_1(T)\|/\|\hat{z}(T)\|}
    \right\|
    +
    \left\|
        \frac{\hat{z}_1(T)}{\|\hat{z}_1(T)\|}
        -
        \frac{v(T)}{\|v(T)\|}
    \right\| \\
    \leq &2\frac{\left\|\tilde{z}_1 - \frac{\hat{z}_1(T)}{\|\hat{z}(T)\|}\right\|}{\left\|\tilde{z}_1\right\|}+
    \left\|
        \frac{\hat{z}_1(T)}{\|\hat{z}_1(T)\|}
        -
        \frac{v(T)}{\|v(T)\|}
    \right\|
    \leq \frac{2\epsilon'}{\|\tilde{z}_1\|} + \frac{2\|v(T)-\hat{z}_1(T)\|}{\|v(T)\|} \leq \frac{2\epsilon'}{\|\tilde{z}_1\|} + 2\delta'\,.
\end{aligned}
\end{equation}
Still by choosing $\epsilon' = \Theta\!\left(\frac{\epsilon}{\sqrt{N}}\right),
    \delta' = \Theta(\epsilon)$, the above error $ \left\|
        \frac{\tilde{z}_1}{\|\tilde{z}_1\|}
        -
        \frac{v(T)}{\|v(T)\|}
    \right\|$ can be bounded by $\mathcal{O}(\epsilon)$.

Combining the QLSA complexity with the amplitude amplification overhead yields
\begin{equation}
  \begin{aligned}
 &\widetilde{O}\left(\underbrace{\frac{\max\|\hat{z}(t)\|}{\|\hat{z}(T)\|}\frac{\alpha_{B_N} \sqrt{T}}{\sqrt{|C_E|}}\log\left(\frac{\|\hat{z}(0)\|+\alpha_{d_N}}{\|\hat{z}(T)\|}\frac{\alpha_{B_N}}{|C_E|\epsilon'}\right)\log(1/\epsilon')}_{\text{From~\eqref{eqn:BNquery}}}\times \underbrace{\sqrt{N}}_{\text{Amplitude amplification}}\right)\\
=        &\mathcal{O}\Biggl(N^2
        \left(
            \alpha_{F_{0,s}}\alpha_Q
            + \alpha_{F_{1,s}}\kappa_Q
            + \frac{\alpha_{F_{2,s}}\kappa_Q^2}{\alpha_Q}
        \right)
        g_v
        \frac{\sqrt{T}}{\sqrt{|C_E|}} \\
        &\times
        \log\!\left(
            \frac{3\sqrt{N}\|v(0)\|+\|E_{0,s}\|}{\|v(T)\|}
            \frac{
                N^{3/2}g_v
                \left(
                    \alpha_{F_{0,s}}\alpha_Q
                    + \alpha_{F_{1,s}}\kappa_Q
                    + \frac{\alpha_{F_{2,s}}\kappa_Q^2}{\alpha_Q}
                \right)
            }{|C_E|\epsilon}
        \right)
        \log(\sqrt{N}/\epsilon)
    \Biggr)
\end{aligned}
\end{equation}
queries to the block-encodings of $F_{i,s}$, $i=0,1,2$,
\begin{equation}
   \mathcal{O}\Biggl(N^2
        \left(
            \alpha_{F_{0,s}}\alpha_Q
            + \alpha_{F_{1,s}}\kappa_Q
            + \frac{\alpha_{F_{2,s}}\kappa_Q^2}{\alpha_Q}
        \right)
        g_v
        \frac{\sqrt{T}}{\sqrt{|C_E|}}
        \log(\sqrt{N}/\epsilon)
    \Biggr)
\end{equation}
queries to the state preparation oracle of $\ket{v(0)}$ and $\ket{d_N}$.

In the equality, we use~\eqref{eq:g_vdef}, the definition of $\alpha_{B_N}$ and $\alpha_{d_N}$, and the choice of $\epsilon'$. Plugging in the dependence of $N$, up to polylogarithmic factors, the calls to coefficient matrices $F_{i,s}$, $i=0,1,2$ become
\begin{equation}
    \widetilde{O}\!\left(
        \sqrt{T}\, \frac{g_v}{\sqrt{|C_E|}}
        \left(
            \alpha_{F_{0,s}}\alpha_Q
            + \alpha_{F_{1,s}}\kappa_Q
            + \frac{\alpha_{F_{2,s}}\kappa_Q^2}{\alpha_Q}
        \right)
        \mathrm{polylog}(T/\epsilon)
    \right)\,.
\end{equation}
Since one query of $B_N$ requires $\mathcal{O}( {\kappa_Q\log(1/\epsilon)})$ calls to the block encoding of $Q$, the number of queries to the block-encoding of $Q$ is
\begin{equation}
    \widetilde{O}\!\left(
         {\kappa_Q} \sqrt{T}\, \frac{g_v}{\sqrt{|C_E|}}
        \left(
            \alpha_{F_{0,s}}\alpha_Q
            + \alpha_{F_{1,s}}\kappa_Q
            + \frac{\alpha_{F_{2,s}}\kappa_Q^2}{\alpha_Q}
        \right)
        \mathrm{polylog}(T/\epsilon)
    \right),
\end{equation}
and the number of queries to the state-preparation oracle for $\ket{v(0)}$ is
\begin{equation}
    \widetilde{O}\!\left(
        \sqrt{T}\, \frac{g_v}{\sqrt{|C_E|}}
        \left(
            \alpha_{F_{0,s}}\alpha_Q
            + \alpha_{F_{1,s}}\kappa_Q
            + \frac{\alpha_{F_{2,s}}\kappa_Q^2}{\alpha_Q}
        \right)
        \mathrm{polylog}(T/\epsilon)
    \right)\,.
\end{equation}
Plugging in the dependence of $\alpha_{F_{i,s}}$ in terms of $\alpha_{F_i}$ and $\alpha_Q$ from~\cref{lem:query0}, we conclude the proof.
\end{proof}

\subsection{Proof of the main theorem}
\label{sec:proofMain}
In this section, we combine the results from the previous sections to prove the main theorem for stable systems, i.e. \cref{thm:main-stable}. The proof of the main theorem for unstable systems, i.e. \cref{thm:main-unstable}, is given in~\cref{appd:unstable}.
\begin{proof}[Proof of \cref{thm:main-stable}]
By applying quantum solvers to the Carleman linearized differential equation, there exists  {a} quantum algorithm that outputs the $\epsilon'$-approximation of the final state, $\ket{\hat{v}_T}$, so that
\begin{equation}
    \left\|\ket{\hat{v}_T}-\frac{v(T)}{\|v(T)\|}\right\|\leq \epsilon'.
\end{equation}
From \cref{prop:stable_linear_solver}, we obtain the query complexity in terms of $\alpha_{F_{i,s}}, g_v$ and the state preparation oracle of $\ket{v(0)}$. Recall that $Q$ is block-encoded with factor $\alpha_Q=\mathcal{O}(\|Q\|)$. We then apply \cref{lem:QLSP} with $A=Q/\alpha_Q$, the state $\ket{Q^{-1}\hat{v}_T}$ with precision $\epsilon''$ can be produced with $\mathcal{O}(\kappa_Q \log(1/\epsilon''))$ calls to the oracle for $\ket{\hat{v}_T}$ and $Q$. We denote the state as $\ket{\hat{u}_T}$, then
\begin{equation}
    \left\|\ket{\hat{u}_T}-\frac{Q^{-1}\hat{v}_T}{\|Q^{-1}\hat{v}_T\|}\right\|\leq \epsilon''.
\end{equation}
Combining the error from the quantum algorithm for linearized equation and the inverse matrix multiplication, we have
\begin{equation}
    \begin{aligned}
        &\left\|\ket{\hat{u}_T} - \frac{u(T)}{\|u(T)\|}\right\|\\
        \leq& \left\|\ket{\hat{u}_T} - \frac{Q^{-1}\hat{v}_T}{\|Q^{-1}\hat{v}_T\|}\right\| + \left\|\frac{Q^{-1}\hat{v}_T}{\|Q^{-1}\hat{v}_T\|} -\frac{Q^{-1}v(T)}{\|Q^{-1}v(T)\|} \right\|\\
        \leq& \left\|\ket{\hat{u}_T} - \frac{Q^{-1}\hat{v}_T}{\|Q^{-1}\hat{v}_T\|}\right\| +\|Q^{-1}\| \left\|\frac{\hat{v}_T/\|\hat{v}_T\|}{\|Q^{-1}\hat{v}_T\|/\|\hat{v}_T\|} -\frac{v(T)/\|v(T)\|}{\|Q^{-1}v(T)\|/\|v(T)\|} \right\|\\
        \leq & \epsilon'' + \|Q^{-1}\|\frac{2\|v(T)\|}{\|u(T)\|}\left\|\frac{\hat{v}_T}{\|\hat{v}_T\|} - \frac{v(T)}{\|v(T)\|}\right\|\leq \epsilon'' + 2\kappa_Q \epsilon'\\
    \end{aligned}
\end{equation}
For $\epsilon_u$-approximation of the state $\ket{u(T)}$, 
we choose $\epsilon' =  {\epsilon_u}/(4\kappa_Q)$ and $\epsilon'' =  {\epsilon_u}/2$. According to~\cref{prop:stable_linear_solver}, there exists a quantum algorithm that output the state approximate $\ket{u(T)}$ up to precision $\epsilon_u$ using
    \begin{equation}
        \widetilde{O}\left(\frac{\kappa_Q\sqrt{T}g_v}{\sqrt{|C_E|}}\alpha_E\polylog(4\kappa_Q/ {\epsilon_u})\right)
    \end{equation}
    queries to the block-encoding of $F_{i,s}$, for $i=0,1,2$, and the state preparation oracle of $\ket{v(0)}$,
    \begin{equation}
        \widetilde{O}\left(\frac{\kappa_Q^{ {2}}\sqrt{T}g_v}{\sqrt{|C_E|}}\alpha_E\polylog(4\kappa_Q/ {\epsilon_u})\right)
    \end{equation}
    queries to the block-encoding of $Q$.

    Finally, we use the final state shift oracle in~\cref{lem:query0}. We first note that the success probability $p=\Theta(1)$ after amplitude amplification. The final state $\ket{x(T)}$ can be obtained from
    \begin{equation}
        \mathcal{O}\left(\frac{\sqrt{\|x(T)-s\|^2+ \|s\|^2 } }{\|x(T)\|}\right)
    \end{equation}
    queries to the state preparation oracle of $\ket{u(T)}$ and $O_s$.

    Combining it with the initial state shift oracle in ~\cref{lem:query0}, we note that we need $\mathcal{O}\left(\frac{\sqrt{\|x(0)\|^2 + \|s\|^2}}{\|x(0)-s\|}\right)$ queries to $O_{x_0}$ and $O_s$ for preparing $U_{\ket{u(0)}}$ and $\mathcal{O}(\kappa_Q)$ queries to $U_{\ket{u(0)}}$ for preparing $U_{\ket{v(0)}}$, there exist a quantum algorithm that approximate the final state $\ket{x(T)}$ up to precision $\epsilon$ using
    \begin{equation}
        \widetilde{O}\left(\frac{\sqrt{\|x(T)-s\|^2+ \|s\|^2 } }{\|x(T)\|}\frac{\kappa_Q\sqrt{T}g_v}{\sqrt{|C_E|}}\alpha_E\polylog\left(\frac{ {8\|u(T)\|}\kappa_Q}{ {\min\{\|u(T)\|, \|x(T)\|\}}\epsilon}\right)\right)
    \end{equation}
    queries to the block-encoding of $ {U_{F_{i}}}$ for $i=0,1,2$, and
     \begin{equation}
        \widetilde{O}\left(\frac{\sqrt{\|x(T)-s\|^2+ \|s\|^2 } }{\|x(T)\|}\frac{\sqrt{\|x(0)\|^2 + \|s\|^2}}{\|x(0)-s\|}\frac{\kappa^2_Q\sqrt{T}g_v}{\sqrt{|C_E|}}\alpha_E\polylog\left(\frac{{8\|u(T)\|}\kappa_Q}{{\min\{\|u(T)\|, \|x(T)\|\}}\epsilon}\right)\right)
    \end{equation}
    queries to $O_{x_0}$ and $O_s$,
    \begin{equation}
        \widetilde{O}\left(\frac{\sqrt{\|x(T)-s\|^2+ \|s\|^2 } }{\|x(T)\|}\frac{\kappa_Q^{ {2}}\sqrt{T}g_v}{\sqrt{|C_E|}}\alpha_E\mathrm{polylog}\left(\frac{{8\|u(T)\|}\kappa_Q}{{\min\{\|u(T)\|, \|x(T)\|\}}\epsilon}\right)\right)
    \end{equation}
    queries to the block-encoding of $Q$. This concludes the proof of \cref{thm:main-stable}.
\end{proof}
The analysis for unstable systems follow the same idea and is postponed to~\cref{appd:unstable}.

\section{Numerical Tests}
\label{sec:example}
We apply the pivot-shifted Carleman (PSC) method on two canonical nonlinear systems: the one dimensional logistic equation and the two dimensional Lotka-Volterra predator-prey equations.
We examine how the Carleman linearization error depends on the truncation order and pivot choice.
{Throughout the numerical experiments, we take $Q = I$ in the Lyapunov preprocessing and set the Carleman rescaling coefficient to $1$.} The source code and  {implementation details} is available in our GitHub repository \cite{repo2026}.

For each model, we assess solution accuracy by computing the $\ell_2$ norm of the error between the simulated trajectory and a reference benchmark.
The closed form analytical solution and the high order Runge-Kutta approximation are used as the reference solutions for the logistic equation and the Lotka-Volterra systems, respectively.

\subsection{Logistic equation}
We begin with the logistic equation, where $x$ is a scalar variable:
\begin{equation}\label{eq:logistic}
\frac{dx}{dt} = x(1-x),\qquad x(0)=x_0.
\end{equation}
We set $x_0 = 0.5$ and evaluate both the standard Carleman linearization and the PSC method over  {$t$ in} $[0, 10]$ across a range of truncation orders. \Cref{fig:logistic}(a)--(d) display the solution trajectories and pointwise errors at $N=4$ and $N=8$.

\begin{figure}[!htbp]
    \centering
    \includegraphics[width=1\linewidth]{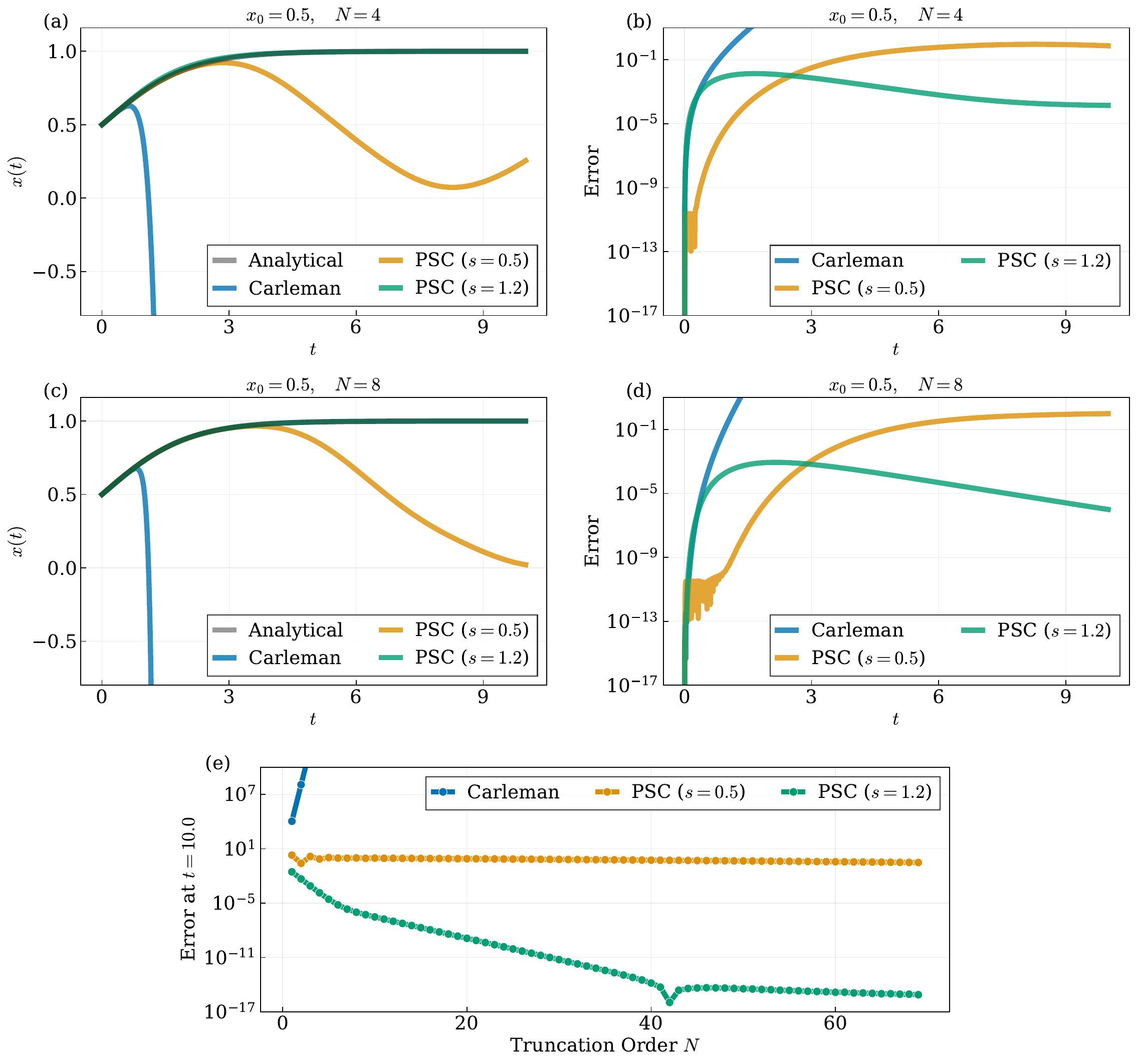}
    \caption{Performance of the standard Carleman linearization and the PSC method on the logistic equation~\cref{eq:logistic} with $x_0 = 0.5$. (a), (c): Solution trajectories of the exact solution, Carleman, PSC with $s=0.5$, and PSC with $s=1.2$, at truncation orders $N=4$ and $N=8$, respectively. (b), (d): Corresponding pointwise errors on a semi-logarithmic scale. The standard Carleman method diverges in both cases due to the violated stability condition $\alpha(F_1) = 1 > 0$. The PSC method with $s=0.5$ also eventually diverges, as this pivot $s$ lies outside the region  {satisfying~\cref{assumption:stability,assumption:stability2}}. The PSC method with $s=1.2$ remains stable and achieves substantially lower error. (e): Error at the final time $t=10$ as a function of truncation order  {$N$ in $\{1,\cdots, 69\}$}. While the errors of Carleman and PSC with $s=0.5$ stagnate or grow with $N$, PSC with $s=1.2$ exhibits exponential convergence, confirming the theoretical analysis.}
    \label{fig:logistic}
\end{figure}

The Carleman solution is unstable and diverges substantially from the analytical solution, as seen in \cref{fig:logistic}(a) and (c). This is because the logistic equation, when cast in the form of \cref{eq:qODEx}, yields $\alpha(F_1) = 1 > 0$, violating the stability condition in \cref{assumption:stability}. Consequently, the Carleman solution departs from the exact solution at approximately $t = 1.2$.

The PSC method addresses this instability by introducing the shifted variable $u = x - s$ and applying Carleman linearization to the resulting system. Rewriting the logistic equation in the form of \cref{eq:MPSquadratic,eq:Fis} gives
\begin{equation}\label{eq:MPSu}
\frac{du}{dt} = -u^2 + (1-2s)u + s - s^2, \qquad u(0) = u_0 := x_0 - s,
\end{equation}
with coefficients $F_{2,s} = -1$, $F_{1,s} = 1-2s$, and $F_{0,s} = s - s^2$. The stability condition \cref{assumption:stability} then requires
\begin{equation}
\alpha(F_{1,s}) = \mu(F_{1,s}) = 1 - 2s < 0,
\end{equation}
i.e., $s > 1/2$. We test two pivots, $s = 0.5$ and $s = 1.2$, to illustrate this dependence. When $s = 0.5$, the stability condition is marginally violated and the PSC solution is expected to diverge over time. \cref{fig:logistic}(a) and (c) confirm this: at $N = 4$ and $N = 8$, the PSC trajectory begins to deviate near $t = 3$, and the error plots in \cref{fig:logistic}(b) and (d) quantify the rapid error growth for long time evolution. \cref{fig:logistic}(e) reinforces this further, showing that the error of PSC with $s = 0.5$ stagnates as $N$ increases, with no meaningful improvement from higher truncation orders.

In contrast, when the pivot is chosen near the stable equilibrium $x^{\star} = 1$, exemplified by $s = 1.2$, the PSC solution closely tracks the analytical solution throughout the entire time interval. \cref{fig:logistic}(a) and (c) show that the trajectory remains accurate at both truncation orders, and the error plots in \cref{fig:logistic}(b) and (d) confirm that the error stays consistently low, in the range of $10^{-3}$ to $10^{-5}$. \cref{fig:logistic}(e) further demonstrates the advantage of this pivot choice: the error at $s = 1.2$ decays exponentially with $N$, confirming that proximity to the stable equilibrium is the key criterion for selecting an effective pivot. The comparison between $s = 0.5$ and $s = 1.2$ highlights the importance of choosing a correct pivot for the PSC method. The former, although closer to the initial condition, fails to achieve high accuracy due to the violation of the stability condition, while the latter, despite being farther from the initial condition, successfully stabilizes the solution and achieves exponential convergence.

\subsection{Lotka-Volterra equations}
We now test Carleman and PSC methods on the normalized two-dimensional Lotka-Volterra predator-prey model~\cite[Equation 37]{novikau2025globalizing},
\begin{equation}\label{eq:LotkaVolterra}
\begin{aligned}
    \partial_t x_1 &= \eta x_1 - \eta x_1x_2\,,\\
    \partial_t x_2 &= -\gamma x_2 + \gamma x_1x_2\,,
\end{aligned}
\end{equation}
where $\eta, \gamma>0$. The  equilibrium points are the unstable saddle $(0,0)$ and the Lyapunov stable  $(1,1)$.

We set the parameters $\eta = 1$ and $\gamma = 0.475$,
the initial condition $x_0 = (0.5, 0.5)$ and evaluate the standard Carleman linearization and the PSC method over $t \in [0, 5]$ for two pivot choices, $s = (0.5, 0.5)$ and $s = (0.7, 0.3)$.
\begin{figure}[!htbp]
    \centering
    \includegraphics[width=1\linewidth]{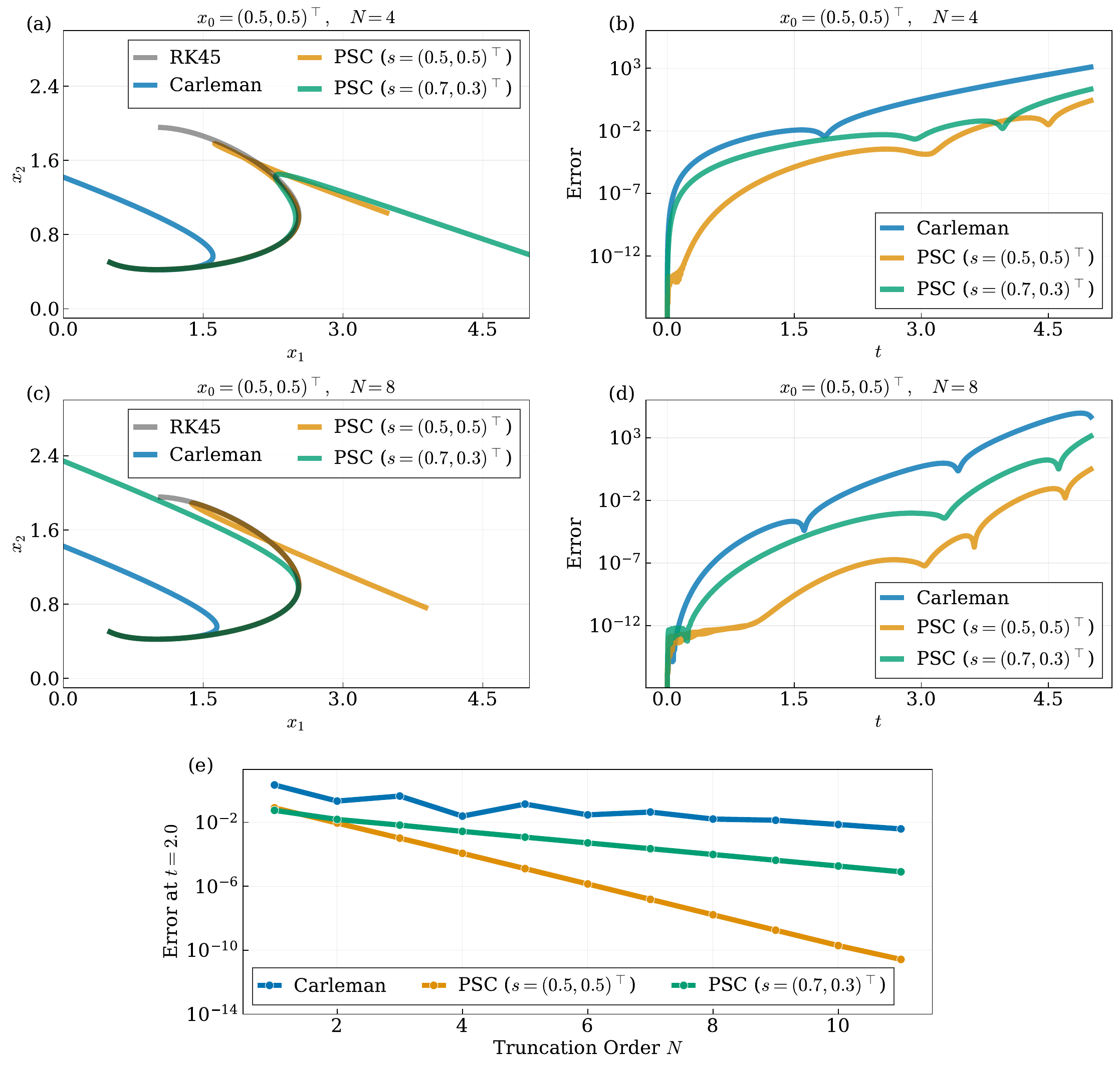}
    \caption{Performance of the standard Carleman linearization and the PSC method on the Lotka-Volterra equations with $x_0 = (0.5, 0.5)$, $\eta=1.0$, and $\gamma=0.475$. The benchmark solution is obtained via the RK45 solver. (a), (c): Phase-plane trajectories of the benchmark solution, Carleman, PSC with $s=(0.5, 0.5)$, and PSC with $s=(0.7, 0.3)$, at truncation orders $N=4$ and $N=8$, respectively. (b), (d): Corresponding $\ell_2$ norm errors over $t \in [0, 5]$ on a semi-logarithmic scale. The standard Carleman method diverges from the reference trajectory, while the PSC method remains stable; in particular, $s=(0.5,0.5)$ achieves substantially lower error at both truncation orders. (e): $\ell_2$ norm error at the  {time} $t=2.0$ as a function of truncation order  {$N=1,\cdots, 11$}. The Carleman error stagnates or grows with $N$, whereas PSC converges consistently as $N$ increases; the pivot $s=(0.5,0.5)$ yields a faster decay rate.}
    \label{fig:LV}
\end{figure}
\cref{fig:LV}~(a) and (b) show the phase-plane trajectories and $\ell_2$ state errors at truncation order $N = 4$. The standard Carleman approximation initially tracks the reference trajectory but diverges after a short time. The PSC method mitigates this instability for both pivots, though the degree of improvement depends on the pivot choice. With $s = (0.7, 0.3)$, the PSC trajectory remains closer to the reference than Carleman but still deviates over time. With $s = (0.5, 0.5)$, the trajectory follows the reference for longer and the error grows more slowly, making it the more effective pivot in this setting.
 Neither pivot produces a shifted system satisfying the stable case assumptions, since $\alpha(F_{1,s}) > 0$ for both choices. The short time guarantee in \cref{prop:unstable_carleman_convergence} applies directly only to $s = x(0) = (0.5, 0.5)$ and  the second pivot is presented as an empirical comparison.
Nevertheless, both PSC variants substantially outperform the standard Carleman method. \cref{fig:LV}(c) and (d) display the corresponding results at $N = 8$ and exhibit the same qualitative behavior.

\cref{fig:LV}(e) complements this picture by reporting the $\ell_2$ error at $t = 2$ as a function of truncation order  {$N \in \{1,\cdots,11\}$},  {a regime not directly covered by our sufficient convergence guarantees}. The standard Carleman method stagnates at small $N$ and  {shows little improvement as $N$ increases.}
 {For both pivot choices, the numerical error decreases with increasing $N$, with $s = (0.5, 0.5)$ exhibiting a faster decay. These results suggest that even when the sufficient stability condition is not satisfied, the Carleman approximation may still improve with increasing truncation order over a sufficiently short time interval. The choice of pivot can affect the corresponding short time convergence region and the observed rate of error decay.}

For illustration, we estimate the theoretical short time lower bound $t^*$ for this model.
The coefficient matrices of \cref{eq:LotkaVolterra} are
\begin{equation}
    F_0 = 0, \quad F_1 = \begin{pmatrix}
        \eta & 0 \\ 0 & -\gamma\\
    \end{pmatrix},\quad  F_2 = \begin{pmatrix}
        0 & -\eta & 0 & 0 \\
        0  & \gamma & 0 & 0 \\
    \end{pmatrix}\,.
\end{equation}
Since $\alpha(F_1) = \eta>0$, the expansion about the origin does not satisfy the stability assumptions required for our long time convergence result. The point $(1,1)$ also does not satisfy~\cref{assumption:stability}.

The shifted coefficient matrices in \cref{eq:MPSquadratic} become $F_{2,s}=F_2$ and
\begin{equation}
    F_{1,s} = \begin{pmatrix}
        \eta - \eta s_2 & -\eta s_1 \\
        \gamma s_2 & -\gamma + \gamma s_1 \\
    \end{pmatrix}, \quad F_{0,s} = \begin{pmatrix}
        \eta (s_1 -  s_1 s_2)\\
        \gamma (s_1 s_2 -  s_2)\\
    \end{pmatrix}.
\end{equation}
When $s = x(0)$, the coefficient matrix $F_{1,s}$ is the Jacobian evaluated at the initial state, the coefficient norms satisfy
\begin{equation}
    \|F_{1,s}\|=\sqrt{\frac{\tau + \sqrt{\tau^2-4\eta^2\gamma^2(s_1+s_2-1)^2}}{2}}\leq \sqrt{\tau} ,\|F_2\| = \sqrt{\eta^2+\gamma^2},\|F_{0,s}\|\leq \sqrt{2(\eta^2 + \gamma^2)}
\end{equation}
where $\tau = \eta^2\left((1-s_2)^2+s_1^2
\right)+\gamma^2(s_2^2 + (s_1-1)^2)\leq 2(\eta^2+\gamma^2)$. Hence, the time limit $t^*$ in~\cref{lm:normboundMPS_2,prop:unstable_carleman_convergence} can be estimated as
\begin{equation}
    t^*  {\geq\frac{1}{e^2(1+2\sqrt{2})\sqrt{\eta^2+\gamma^2}}}.
\end{equation}

\subsection{Pivot sensitivity and computational cost}

 {We further examine how the accuracy of PSC depends on the choice of pivot. This test complements \cref{fig:logistic} and \cref{fig:LV}, where only a small number of representative pivots were shown. For each pivot and truncation order, we record the maximum error over the simulated time interval. We also mark a run as a failure if the ODE solver does not report success, if non-finite values occur, or if the maximum-in-time error exceeds $100$.}

 {For the logistic equation, the pivot is used to create a stable shifted system. Recall that from the stability condition \cref{assumption:stability}, we require $s > 0.5$, where the shifted linear coefficient changes sign. For $s\leq 0.5$ the shifted linear part is not dissipative, while for $s>0.5$ it becomes dissipative. We therefore sweep $s\in[-0.5,1.5]$. This interval contains the origin expansion $s=0$, the initial-condition pivot $s=x_0=0.5$, the stable equilibrium $x^\star=1$, and the representative stable-equilibrium pivot $s=1.2$ used in \cref{fig:logistic}. The purpose of the sweep is to show how the error changes as the pivot crosses the stability threshold, and then moves toward and beyond the stable equilibrium.}

 {For the Lotka-Volterra equation, the pivot plays a different role. The pivots used in \cref{fig:LV} do not put the system in the stable shifted regime, so the relevant theory is the short-time result. In that case, the natural pivot is the initial condition $s=x(0)$, because it gives $u(0)=0$ and suppresses the initial contribution of the omitted high-order term. The sensitivity test therefore asks how robust this short-time advantage is when the pivot is moved away from $x_0=(0.5,0.5)$. We sweep the rectangle $[0,1.2]\times[0,1.0]$, which stays in the positive quadrant and includes the origin pivot, the theory-motivated pivot $s=x_0$, the comparison pivot $(0.7,0.3)$ from \cref{fig:LV}, and pivots closer to the coexistence equilibrium $(1,1)$. This range lets us compare pivots near the initial condition with pivots farther from the early-time trajectory.}

\begin{figure}[t]
    \centering
    \includegraphics[width=0.5\linewidth]{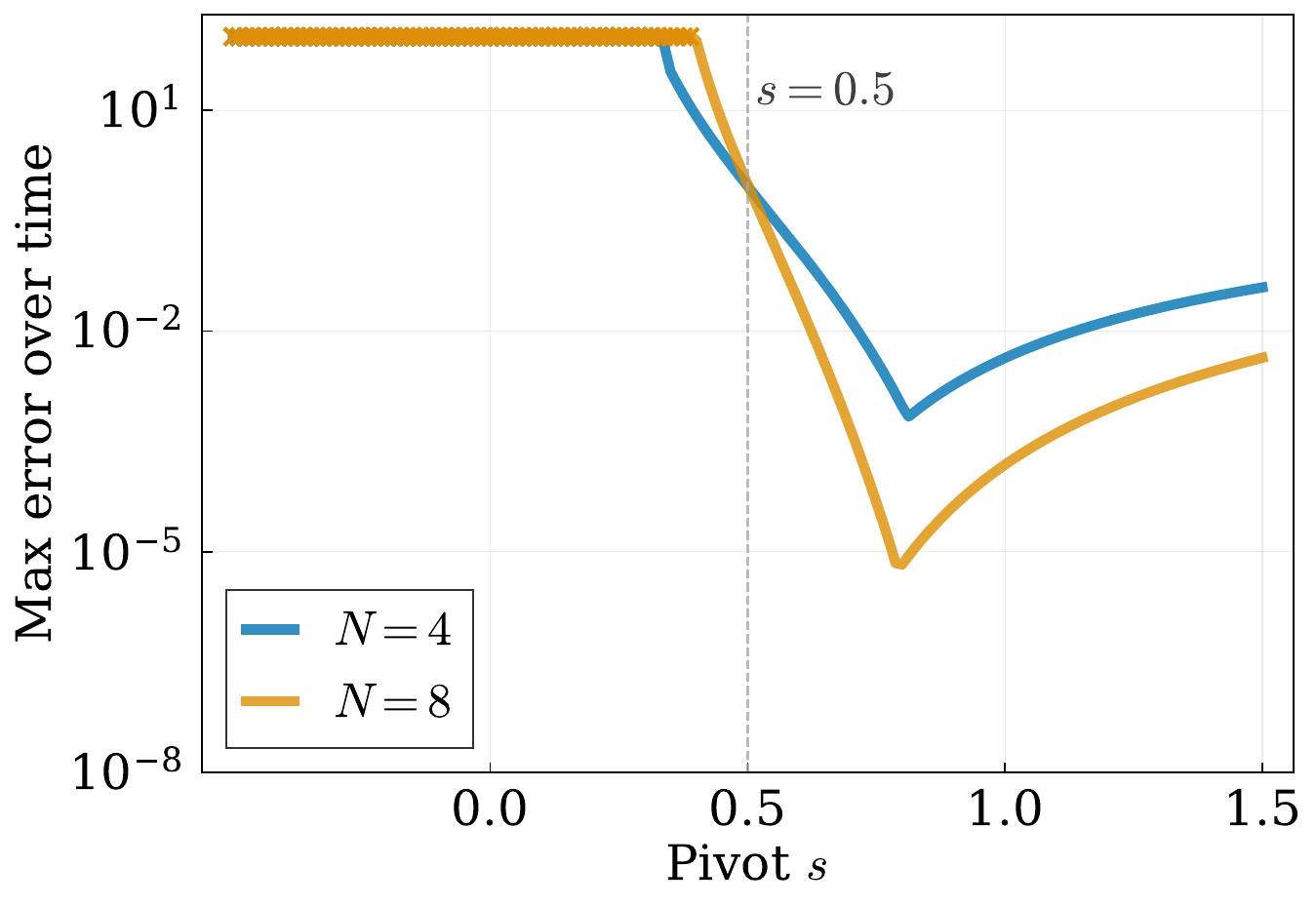}
    \caption{Pivot sensitivity of the PSC approximation for logistic equation, showing maximum-in-time and final-time error as functions of the scalar pivot. Displayed errors are clipped at $10^2$ for readability, and crosses indicate pivots flagged by the failure criterion. }
    \label{fig:logistic-pivot-sensitivity}
\end{figure}

 {The logistic sweep is shown in \cref{fig:logistic-pivot-sensitivity}. The largest change occurs when the pivot crosses $s=0.5$. For $s\leq 0.5$, the shifted linear coefficient is nonnegative, and the numerical error grows rapidly. Once $s>0.5$, the shifted linear part becomes dissipative and the finite-time error drops substantially, even before the sufficient long-time condition is met. The smallest error on this finite grid occurs near $s=0.8$, where the maximum error is
$6.691\times 10^{-6}$ at $N=8$. The pivot $s=1.2$, used in \cref{fig:logistic} because it is close to the stable equilibrium, also remains accurate over the whole interval, with maximum error $8.818\times 10^{-4}$ at $N=8$. Thus the sensitivity test supports the mechanism described above: shifting past $s=0.5$ stabilizes the linear part, while pivots near the stable equilibrium continue to give robust accuracy. The same sweep also identifies a clear failure region for pivots far to the left of the trajectory; under the stated criterion, 72 of the 161 tested pivots are flagged at $N=8$. The $N=4$ sweep shows the same qualitative pattern, with a sharp improvement after the pivot passes $s=0.5$ and its smallest error attained in the stable shifted region.}

\begin{figure}[t]
    \centering
    \includegraphics[width=0.95\linewidth]{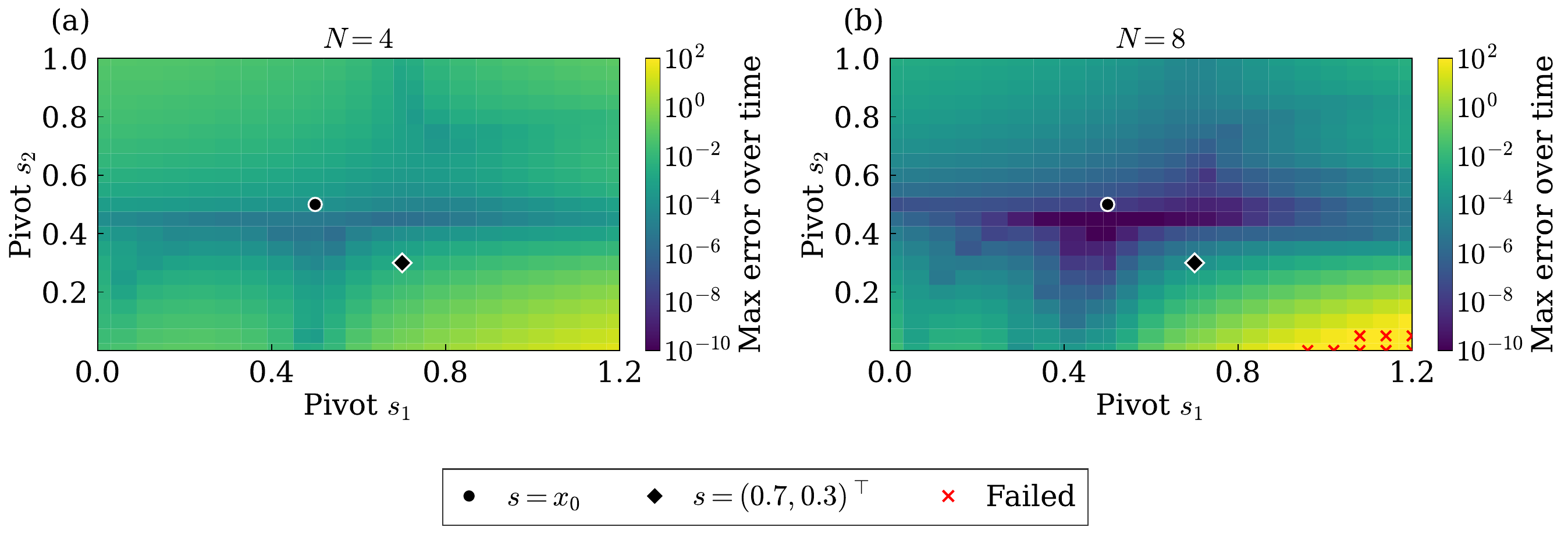}
    \caption{Lotka-Volterra heatmaps of $\log_{10}$ maximum-in-time error over the two-dimensional pivot grid, plotted with a common color scale for the two truncation orders. Black markers indicate the pivots used in the main comparison figures, and red crosses indicate failed pivots.
    }
    \label{fig:Lotka-Volterra-pivot-sensitivity}
\end{figure}

 {The Lotka-Volterra heatmaps in \cref{fig:Lotka-Volterra-pivot-sensitivity} show the two-dimensional dependence of the approximation error on the pivot. The pivot $s=x_0=(0.5,0.5)$, used in \cref{fig:LV}, gives maximum error $1.643\times 10^{-8}$ at $N=8$. This agrees with the short-time theory, since the choice $s=x(0)$ gives $u(0)=0$ and makes the omitted high-order term small near the initial time. The heatmaps also show that this behavior is not confined to a single point: nearby pivots form a low-error region on this fixed time interval. The best tested grid pivot is $s=(0.6,0.45)$, with maximum error $4.743\times 10^{-11}$. We interpret this nearby improvement as a robustness effect for the finite-time experiment, rather than as a separate pivot selection rule. In contrast, pivots farther from the early-time trajectory can degrade the approximation. In the tested grid, 8 of the 483 Lotka-Volterra pivots are flagged at $N=8$, concentrated in the lower-right corner of the pivot domain. The $N=4$ heatmap displays the same low-error region around the initial-condition pivot, although with the larger errors expected at lower truncation order.}

\section{Discussions and open questions}
\label{sec:discussion}

In this work, we proposed and analyzed a pivot-shifted Carleman (PSC) method for stabilizing Carleman linearization in the solution of nonlinear ordinary differential equations. We established theoretical guarantees for the stability and convergence of the PSC method and demonstrated its effectiveness through numerical examples. Our results reveal two complementary aspects of the method. First, for stable systems, when the pivot point is chosen appropriately, the PSC method can efficiently stabilize the original Carleman linearization and achieve exponential convergence with respect to the truncation order. This modification significantly broadens the regime in which Carleman-based quantum algorithms can be applied efficiently.  {At the same time, identifying a pivot that satisfies the required stability condition may require prior structural information about the underlying dynamics. In particular, choosing the pivot at a stable equilibrium requires prior knowledge of that equilibrium. Developing efficient classical or quantum procedures for selecting suitable pivots without full information of the dynamics remain an interesting direction for future work. } Second, for unstable systems, even when the stability condition is violated, we choose the pivot to be the initial condition, $s=x(0)$, so that the shifted variable satisfies $u(0)=0$. This shrinks the standard Carleman truncation factor $\|u\|^{N+1}$ near the initial time and yields improved short-time accuracy over standard Carleman linearization, with the error continuing to decrease as the truncation order grows. This suggests that the PSC method can be beneficial even when the stability condition is not strictly satisfied or when no prior knowledge of a stable equilibrium is available. These convergence properties also provide a theoretical foundation for the single-pivot step underlying the empirical schemes in~\cite{endo2024divergence,novikau2025globalizing}.

There are several open questions and future directions arising from our work. First, although we have established the stability and convergence of the PSC method, the optimal choice of pivot point remains open. Developing systematic methods for selecting pivots that maximize stability and improve convergence would be valuable.
Second, while our analysis focuses on quadratic nonlinearities, extending the PSC method to higher-order nonlinearities and more general classes of systems is straightforward at the algorithmic level. However, the corresponding stability and convergence analysis in these more complex settings may require new techniques and insights.
Third, in the unstable case, our results show that the PSC method can achieve short time accuracy even when the stability condition is violated. As demonstrated in~\cite{novikau2025globalizing}, long time simulation may also be possible by switching pivots multiple times. However, to the best of our knowledge, it remains unclear how to implement multiple pivot switching efficiently in practice, and the stability and convergence properties of such an approach are not yet well understood.
Our work opens the door to addressing these questions and further advancing the theory and applications of Carleman linearization in both classical and quantum settings.

\section*{Acknowledgements}
This work was supported in part by the University of Michigan through a startup grant of Z.D. (Z.D., K.W.) and the Van Loo Postdoctoral Fellowship (K.W.).

\section*{Author Contributions}
All authors participated in the research and prepared the manuscript. K.W., Z.D. analyzed the methods, Z.J. wrote the code. Z.D. provided the initial idea. S.V. and Z.D. led the project.  
Generative AI tools (ChatGPT 6.0, 6.2) were used in proofreading, formatting and checking results in the manuscript.
\bibliographystyle{quantum}
\bibliography{ref}

\appendix
\section{Oracle construction for the shifted system (Proof of \texorpdfstring{\cref{lem:query0}}{PDFstring})}\label{appd:QuantumOracles}
In this section, we provide a more detailed discussion of~\cref{lem:query0}. We assume block encoding oracles for the original coefficient matrices $F_{0}, F_1, F_2$ with normalization coefficient $\alpha_{F_{0}}, \alpha_{F_1}$ $\alpha_{F_2}$, as well as the state preparation oracle for the pivot $s$.
 The shifted coefficient matrices follow
\begin{equation}
    F_{2, s} = F_{2}, \quad F_{1, s}= F_1 + F_2(s\otimes I_n + I_n\otimes s), \quad F_{0,s} = F_2s^{\otimes 2}+ F_1 s + F_0\,.
\end{equation}
Then $\alpha_{F_{2,s}} = \alpha_{F_2}$.
\begin{lem}[block encoding of the tensor product]
     Suppose $U_A, U_B$ are the $(\alpha_1, a,\epsilon)$-block encoding of $A$ and $(\alpha_2, b,\epsilon)$-block encoding of $B$, then there exists $(\alpha_1\alpha_2, a+b, \epsilon)$-block encoding of $A\otimes B$ by composing $U_A\otimes U_B$ with $\mathcal{O}(a+b)$ SWAP gates.
\end{lem}
\begin{lem}(block encoding of the multiplications~\cite[Proposition 9.12]{lin2026quantum})
\label{lem:be_multi}
    Suppose $U_A, U_B$ are the $(\alpha_1, a,\epsilon_1)$-block encoding of $A$ and $(\alpha_2, b,\epsilon_2)$-block encoding of $B$, then there exists $(\alpha_1\alpha_2, a+b, \alpha_1\epsilon_2 + \alpha_2 \epsilon_1 + \epsilon_1 \epsilon_2)$-block encoding of $AB$.
\end{lem}
Combined with LCU, we know that $I_n$ is the $(1,0,0)$ -block encoding of itself, then there exists the $(2\|s\|, \log(2),0)$ block encoding of $s \otimes I_n + I_n \otimes s$.
Thus, $F_{1,s}$ admits an $(\alpha_{F_1}+2\|s\| \alpha_{F_2}, a_{F_2} + 2\log(2), (2\|s\|+1)\epsilon)$-block encoding. Similarly, $F_{0,s}$ admits an $(\alpha_{F_0}+\alpha_{F_1} \|s\| + \alpha_{F_2}\|s\|^2, a + \log(4), (\|s\|^2  + \|s\|+1) \epsilon)$-block encoding.

 {With the condition number $\kappa_Q = \|Q\|\|Q^{-1}\|$, and the normalization factor of a block encoding of $Q$ as $\alpha_Q = \Theta(\|Q\|)$. Using quantum singular value transformation~\cite{gilyen2019quantum} to approximate the reciprocal function, a block encoding of $Q^{-1}$ can be constructed with normalization factor $\alpha_{Q^{-1}} = \Theta(\|Q^{-1}\|) = \Theta(\kappa_Q/\alpha_Q)$, using $\mathcal{O}(\kappa_Q\log(\kappa_Q/\epsilon))$ queries to the block encoding of $Q$ and its adjoint.
Applying the block-encoding multiplication rule in~\cref{lem:be_multi} to the transformed coefficient matrix
\begin{equation*}
    E_{2,s} = Q F_{2,s} (Q^{-1}\otimes Q^{-1}),
\end{equation*}
we obtain a block encoding with normalization factor
\begin{equation}
    \alpha_{E_{2,s}} =\Theta\left(\alpha_{F_{2,s}}\alpha_Q \left(\frac{\kappa_Q}{\alpha_Q}\right)^2\right) = \Theta\left(\frac{\alpha_{F_{2,s}}\kappa_Q^2}{\alpha_Q}\right)\,.
\end{equation}
Since $E_{2,s}$ contains two factors of $Q^{-1}$, its implementation requires two applications of the QSVT block encoding of $Q^{-1}.$ Therefore, the corresponding number of queries to the block encoding of $Q$ remains $\mathcal{O}(\kappa_Q \log(1/\epsilon))$. It can be upper bounded by  $\mathcal{O}(\kappa_Q \log(1/\epsilon))$ queries.
Similarly, we can derive $\alpha_{E_{0,s}} = \alpha_{F_{0,s}}\alpha_Q$, $\alpha_{E_{1,s}} = \alpha_{F_{1,s}}\kappa_Q$. }

We also provide detailed discussion to prepare the initial state as the input of the QLSA based quantum algorithm for linear differential equations and shift the final state back to the original differential equation.
\paragraph{Preparation of shifted states}

Finally, we discuss the complexity of preparing shifted states. Given a vector $\psi$ with known norm $\|\psi\|$, and the associated normalized quantum state $\ket{\psi}$, we assume there exists an oracle $U_{\psi}$ that prepares $\ket{\psi}$ in the form
\begin{equation}\label{eqn:stateprep}
U_{\psi}\ket{0^{m+1}}
=
\sqrt{p}\ket{0}\ket{\psi}
+
\sqrt{1-p}\ket{1}\ket{\Phi},
\end{equation}
where the success probability $p$ is known. By post-selecting the first register on the outcome $\ket{0}$, we obtain the state $\ket{0}\ket{\psi}$. Our goal is to prepare the normalized state associated with the shifted vector $\psi-s$, namely,
\[
\ket{\psi-s}:=\frac{\psi-s}{\|\psi-s\|},
\]
using the oracles $O_s$ and $U_\psi$. This setting covers both the initial shift and the final shift in \cref{lem:query0}:
\begin{itemize}
  \item For the initial shift, we take $U_{\psi}=O_{x_0}$, and in this case $p=1$.
  \item For the final shift, $U_{\psi}$ is constructed from the QLSA based linear differential equation solvers, we take $U_\psi = O_{u(T)}$,  with the corresponding probability $p $. In this case, we aim to prepare the state $\ket{x(T)} = \ket{u(T) + s}$, which only changes the following construction slightly with the same query complexity scaling.
\end{itemize}

We prepare the state $\ket{\psi-s}$ using one additional ancilla qubit. First, by combining a rotation on the ancilla with controlled-$U_\psi$, we can prepare
\[
\frac{\sqrt{p}\|\psi\|}{\sqrt{\|\psi\|^2+\|s\|^2}}\ket{0}\ket{0}\ket{\psi}
+
\frac{\sqrt{p}\|s\|}{\sqrt{\|\psi\|^2+\|s\|^2}}\ket{1}\ket{0^{m+1}}
+
\frac{\sqrt{1-p}\|\psi\|}{\sqrt{\|\psi\|^2+\|s\|^2}}\ket{0}\ket{1}\ket{\Phi}
+
\frac{\sqrt{1-p}\|s\|}{\sqrt{\|\psi\|^2+\|s\|^2}}\ket{1}\ket{1}\ket{0^m}.
\]
Next, applying controlled-$O_s$ to the last register conditioned on the first ancilla being $\ket{1}$ gives
\[
\sqrt{p}\left(
\frac{\|\psi\|}{\sqrt{\|\psi\|^2+\|s\|^2}}\ket{0}\ket{0}\ket{\psi}
+
\frac{\|s\|}{\sqrt{\|\psi\|^2+\|s\|^2}}\ket{1}\ket{0}\ket{s}
\right)
+\ket{\perp},
\]
where $\ket{\perp}$ denotes terms orthogonal to the subspace in which the first two registers are $\ket{0}\ket{0}$ or $\ket{1}\ket{0}$.

Applying a Hadamard gate to the first register yields
\[
\frac{\sqrt{p}}{\sqrt{2(\|\psi\|^2+\|s\|^2)}}
\ket{0}\ket{0}\bigl(\|\psi\|\ket{\psi}-\|s\|\ket{s}\bigr)
+\ket{\perp}.
\]
Since $\|\psi\|\ket{\psi}=\psi$ and $\|s\|\ket{s}=s$, this becomes
\[
\frac{\sqrt{p}\,\|\psi-s\|}{\sqrt{2(\|\psi\|^2+\|s\|^2)}}\ket{0}\ket{0}\ket{\psi-s}
+\ket{\perp}.
\]
Therefore, by amplitude amplification followed by post-selection on the first two registers being $\ket{0}\ket{0}$, we can prepare $\ket{\psi-s}$ using
\[
\mathcal{O}\!\left(
\frac{\sqrt{\|\psi\|^2+\|s\|^2}}{\sqrt{p}\,\|\psi-s\|}
\right)
\]
queries to $U_\psi$, $U_\psi^\dagger$, $O_s$, and $O_s^\dagger$.

\section{Short time convergence for unstable systems}
\label{appd:unstable}
In this section, we provide the proofs for the unstable systems. Throughout the proof \textit{we choose $Q = I$ and pivot $s = x(0)$, it reveals that $v(0) = 0$ and $E_{i,s} = F_{i,s}$ for all $i = 0, 1, 2$.} The error from the algorithm are mainly from Carleman truncation error and the error from the QLSA based quantum algorithm for linear differential equations.

\subsection{Convergence of the Carleman linearization}
We first prove~\cref{prop:unstable_carleman_convergence} using the following bound for the solution $v(t)$. Set
\begin{equation}
     {S : = \|F_{2,s}\| + \|F_{1,s}\| + \|F_{0,s}\|.}
\end{equation}
If $S=0$, both $v(t)$ and the truncation error are identically zero. We therefore assume $S>0$.

\begin{lem}[General system: short time bound of $\|v(t)\|$]\label{lm:normboundMPS_2}
The solution $v(t)$ of~\cref{eqn:v_t_equation} satisfies
\[
\|v(t)\|\leq  {S}t\leq 1,
\qquad 0\leq t\leq t':= {S^{-1}}.
\]
In particular, $\|v(t)\|<1$ for $0\leq t<t'$.
\end{lem}
With ~\cref{lm:normboundMPS_2}, the norm of the solution $\|v(t)\|$ can be bounded by one for any time $t<t'$. Therefore, we can prove the convergence of Carleman truncation accordingly.
\begin{proof}[Proof of~\cref{prop:unstable_carleman_convergence}]
The error vector $\xi(t)$ satisfies the differential equation in~\cref{eq:CarlemanLinearized},
\begin{equation}
    \partial_t \xi = B_N \xi +\hat d_N\,, \quad \xi(0) = 0.
\end{equation}
From Duhamel's principle, the error vector can be written using the integral
\begin{equation}
    \xi(t) = \int_0^t e^{B_N(t-\tau)}\hat{d}_N(\tau)\mathrm{d}\tau,
\end{equation}
whose norm is upper bounded by
\begin{equation}
    \|\xi(t)\|\leq \max_{\tau\in[0,t]}\|e^{B_N\tau}\|\|B_{N,N+1}\|\int_0^t\|v(\tau)\|^{N+1}\mathrm{d}\tau.
\end{equation}
Let $\Lambda = \sum_{j=1}^N\ketbra{j}\otimes \frac{1}{j} I_{n^j}$.
Similar to the proof in \cref{eq:bd_BN}, we replace $B_N$ with $\Lambda B_N$, then
\begin{equation}
    \begin{aligned}
        &\xi^\dagger(\Lambda B_N  +B_N^\dagger\Lambda
        )\xi  \\
        &= \sum_j \frac{1}{j}\xi_j^\dag (B_{j,j} + B_{j,j}^\dag)\xi_j + \sum_{j} \xi_{j+1}^\dag \left(\frac{1}{j+1}B_{j+1, j} + \frac{1}{j}B_{j, j+1}^\dag\right)\xi_j + \sum_j \xi_j^\dag\left(\frac{1}{j+1}B_{j+1,j}^\dag + \frac{1}{j}B_{j,j+1}\right)\xi_{j+1} \\
        &  \leq \sum_j \frac{2}{j}\mu(B_{j,j})\|\xi_j\|^2 + \sum_j\frac{2}{j+1}\|B_{j+1,j}\|\|\xi_j\|\|\xi_{j+1}\|+\sum_j\frac{2}{j}\|B_{j,j+1}\|\|\xi_j\|\|\xi_{j+1}\|\\
        &\leq \sum_j 2\mu(F_{1,s})\|\xi_j\|^2 + 2\|F_{0,s}\|\|\xi_j \|\| \xi_{j-1}\| + 2\|F_{2,s}\| \|\xi_j \|\|\xi_{j-1}\|\\
        & \le \xi_{\tilde{G}}^\dag (\tilde{G}+ \tilde{G}^\dag)\xi_{\tilde{G}}\leq 2\mu(\tilde{G})\|\xi_{\tilde{G}}\|^2\,,
    \end{aligned}
\end{equation}
where $\tilde{G}\in\mathbb{C}^{N\times N}$ is a tridiagonal matrix with $\tilde{G}_{j,j} = \mu(F_{1,s}), \tilde{G}_{j-1,j} = \|F_{0,s}\|, \tilde{G}_{j+1,j} = \|F_{2,s}\|.$
Then by Gershgorin's circle theorem,
\begin{equation}
      {\frac{1}{2}}\lambda_{\max}(\Lambda B_N + B_N^\dag \Lambda) \leq \mu(F_{1,s}) + \|F_{0,s}\| + \|F_{2,s}\|.
\end{equation}
The log norm becomes
\begin{equation}
\begin{aligned}
&\mu_\Lambda(B_N)
= \frac12 \max_{x\ne 0}
\frac{\langle x,(\Lambda B_N+B_N^\dagger \Lambda)x\rangle}{\langle x,\Lambda x\rangle}
\leq   {N}\max\{\mu(F_{1,s})+\|F_{0,s}\|+\|F_{2,s}\|,  {0} \}
\end{aligned}
\end{equation}
by \cite[Lemma 3.31]{plischke2005transient}.  By \cite[Theorem 3.35]{plischke2005transient}, we have $\|e^{B_Nt}\|\leq \sqrt{\kappa_{\Lambda} }e^{\mu_\Lambda (B_N)t}$, then
\begin{equation}\label{eq:boundCB_N}
    \|e^{B_Nt}\| \leq  \sqrt{N}\max\{e^{ {N}\left(\mu(F_{1,s}) + \|F_{0,s}\| + \|F_{2,s}\|\right)t}, 1 \}.
\end{equation}
From~\cref{lm:normboundMPS_2}, $\|v(t)\|< 1$ for all $t<t'=\frac{1}{S}$. When $\mu(F_{1,s}) + \|F_{0,s}\| + \|F_{2,s}\|\geq 0$, the Carleman truncation error $\xi$ can be bounded as
\begin{equation}
\begin{aligned}
    &\|\xi(t)\|
    \leq   t \sqrt{N}\sup_{\tau\in[0,t]}e^{N\left(\mu(F_{1,s}) + \|F_{0,s}\| + \|F_{2,s}\|\right)t} \|B_{N,N+1}\|\sup_{\tau\in[0,t]}\|v(\tau)\|^{N+1}\\
    &\le tN^{3/2}\|F_{2,s}\| \left(tSe^{\left(\mu(F_{1,s}) + \|F_{0,s}\| + \|F_{2,s}\|\right)t}\right)^{N}\le tN^{3/2}\|F_{2,s}\| \left(tSe^{St}\right)^{N} \,,
\end{aligned}
\end{equation}
where we use the fact that $\|B_{N,N+1}\| \leq N\|F_{2,s}\|$.
For all $t<t^*=\frac{t'}{e^2}=\frac{1}{e^2S}$, $tSe^{St}\le
St^*e^{St^*}\le\frac{1}{e}<1,$ therefore, the error $\xi(t)$ decays exponentially with respect to the truncation order $N$.
The Carleman error can be bounded as $\|\xi(t)\|\leq \epsilon \|v(t)\|$ with the truncation order
\[
N =
{\Theta\left(\log\left(\frac{t\|F_{2,s}\|}{\|v(t)\|\epsilon}\right)\right).}
\]
When $\mu(F_{1,s}) + \|F_{0,s}\| + \|F_{2,s}\|< 0$, the upper bound becomes
\begin{equation}
    \|\xi(t)\|\leq   tN^{3/2}\|F_{2,s}\| \left(t S\right)^{N}\leq tN^{3/2}\|F_{2,s}\| \left(\frac{1}{e^2}\right)^N\,.
\end{equation}
This completes the proof of~\cref{prop:unstable_carleman_convergence}.
\end{proof}
The details in proving the short time bound on  the norm of $v(t)$ are provided as follows.
\begin{proof}[Proof of~\cref{lm:normboundMPS_2}]
Since $\|v(0)\|=0$, there exists a short time interval so that $\|v(t)\|< 1$.
On the interval where $\|v(t)\|\leq 1$, the integral equation yields
\begin{equation}
    \|v(t)\|
    \le \int_0^t\|F_{2,s}\|\|v(\tau)\|^2  + \|F_{1,s}\|\|v(\tau)\| +\|F_{0,s}\|\mathrm{d}\tau \leq St.
\end{equation}
If the solution first reached norm $1$ at some $t<S^{-1}$,
this estimate would give $1\leq St<1$, which is a contradiction. Therefore, for all $t < t'= \frac{1}{S}$, we have $\|v(t)\|<1.$
\end{proof}

\subsection{Quantum solvers for linear differential equations}
Next, we incorporate the error from the QLSA based quantum solver for linear differential equations with the corresponding query complexity.
\begin{proof}[Proof of~\cref{prop:unstable_linear_solver}]
For systems that remain unstable after pivot shift, we apply the truncated Taylor series linear ODE solvers of~\cite[Theorem 7]{kroviImprovedQuantumAlgorithms2023} to the Carleman linearized equation. The linear system used by this solver differs slightly from \cref{eq:AMMp} by carrying an additional register for the Taylor index.  
 {To keep the block encoding input model, we use the block-encoding version of this result in~\cite[Theorem 4.26]{dutt2026simulating}.}

Follow the same idea in the proof of~\cref{prop:stable_linear_solver},
we first consider the linear differential equation from truncated Carleman linearization,
\[
\partial_t \hat{z} = B_N \hat{z} + d_N\,.
\]
Denote its $\epsilon'$-output from quantum solver as $\tilde{z}$, which approximates the final state as
\begin{equation}
        \left\|\tilde{z} -\ket{\hat{z}(T)}\right\|\leq \epsilon'\,.
\end{equation}
From the proof of~\cref{prop:unstable_carleman_convergence}~\cref{eq:boundCB_N},
    \[\|e^{B_Nt}\| \leq  \sqrt{N}e^{ {N}\left( {\max\{\mu(F_{1,s}) + \|F_{0,s}\| + \|F_{2,s}\|,0\}}\right)t}\leq \sqrt{N}e^{N/e^2}\,,\quad t\leq T<t^*\,.
    \]
We notice that the normalization coefficient of the block encoding oracle for $B_N$ satisfies $\alpha_{B_N} = N\sum_{i = 0}^2 \alpha_{F_{i,s}}$.
The norm of $B_N$ can be bounded as $\|B_N\| \leq \alpha_{B_N}$. The algorithm, therefore, requires
\begin{equation}\label{eq:QLSAcomplexity2}
\widetilde{\mathcal{O}}\left(T
\frac{\max_{t\in[0,T]}\|\hat{z}(t)\|}{\|\hat{z}(T)\|} \sum_{i = 0}^2\alpha_{F_{i,s}} \poly(N) e^{N/e^2}
\poly\left(
\log\left(1+\frac{T e^2 \|F_{0,s}\|}{\|\hat{z}(T)\|}\right),
\log\!\left(\frac{1}{\epsilon'}\right)\right)
\right)
\end{equation}
queries to the oracles for $B_N, \ket{d_N}$. Note that since $s=x(0)$, we have $v(0)=\vec{0}$, so no initial state preparation is needed.

From~\cref{prop:unstable_carleman_convergence}, when $T< t^* =\frac{1}{e^{ {2}}(\|F_{2,s}\| + \|F_{1,s}\| + \|F_{0,s}\|)}$,
the relative error from Carleman truncation can be bounded by $\delta'$, $\|\xi(T)\|\leq \delta'\|v(T)\|$ by choosing
\begin{equation}\label{eqn:N_lower_bound_unstable}
N = \Theta\left(\log\left(\frac{T\|F_{2,s}\|}{\|v(T)\|\delta'}\right)\right).
\end{equation}

By measuring the registers, we extract the first block of $\tilde z$, denoted as $\tilde{z}_1$, and use it to approximate $\ket{u(T)}$.
Following the same approximation in ~\cref{eq:z1tilde,eq:g_vdef,eq:z1tilde_error}, we choose $\epsilon' = \Theta\!\left(\frac{\epsilon}{\sqrt{N}}\right),
\delta' = \Theta(\epsilon),$ so that
\[
\left\| \frac{\tilde{z}_1}{\|\tilde{z}_1\|}- \frac{v(T)}{\|v(T)\|}
    \right\|\leq \frac{2\epsilon'}{\|\tilde z_1\|} + 2\delta'\leq \mathcal{O}(\epsilon)
\]
    The success probability can be boosted to a constant by $\mathcal{O}(\sqrt{N})$ rounds of amplitude amplification.

Furthermore, the parameters in the query complexity can be approximated explicitly as the solution ratio coefficient
    \begin{equation}
    \frac{\max_{t\in[0,T]}\|\hat{z}(t)\|}{\|\hat{z}(T)\|}\leq \sqrt{N} \frac{1+\delta'}{1-\delta'} \frac{\max_{t\in[0,T]}\|v(t)\|}{\|v(T)\|} = \sqrt{N}\frac{1+\delta'}{1-\delta'} g_v,
\end{equation}
    with $g_v = \frac{\max_{t\in[0,T]}\|v(t)\|}{\|v(T)\|}$.

    Combining the complexity from the quantum algorithm with the amplitude amplification overhead, for time $T<t^*$, we have an $\epsilon$-approximation of $\ket{v(T)} $ using
        \begin{equation}
        \adjustbox{max width=\linewidth}{$\displaystyle
        \begin{aligned}
            &\widetilde{\mathcal{O}}\left(\underbrace{T
\frac{\max_{t\in[0,T]}\|\hat{z}(t)\|}{\|\hat{z}(T)\|} \sum_{i = 0}^2\alpha_{F_{i,s}} \poly(N) e^{N/e^2}
\poly\left(
\log\left(1+\frac{T e^2 \|F_{0,s}\|}{\|\hat{z}(T)\|}\right),
\log\left(\frac{1}{\epsilon'}\right)\right)}_{\text{From }\cref{eq:QLSAcomplexity2}}
\times \underbrace{\sqrt{N}}_{\text{Amplitude Amplification}}\right)\\
& = \widetilde{\mathcal{O}}\left(T
\frac{1+\delta'}{1-\delta'}g_v \sum_{i = 0}^2\alpha_{F_{i,s}} \poly(N) e^{N/e^2}
\poly\left(
\log\left(1+\frac{T e^2 \|F_{0,s}\|}{(1-\delta')\|v(T)\|}\right),
\log\left(\frac{1}{\epsilon'}\right)\right)\right) \\
&=\widetilde{\mathcal{O}}\left(T^2g_v\frac{\sum_{i = 0}^2\alpha_{F_{i,s}}\left\|F_{2,s}\right\|}{\|v(T)\|\epsilon}\right)
        \end{aligned}
        $}
    \end{equation}
    queries to the oracle for $B_N$, $\ket{d_N}$.
\end{proof}

\begin{proof}[Proof of~\cref{thm:main-unstable}]
We follow the steps used in the proof of~\cref{thm:main-stable}.
Since $Q=I$ and $s=x(0)$, we have $u = v$ and $v(0) = 0$.

By~\cref{lem:query0}, the final state $\ket{x(T)}$ can be prepared
from $\ket{u(T)}$ and $\ket{s}=\ket{x(0)}$ using
\[
\mathcal{O}\left(
    \frac{\sqrt{\|x(T)-s\|^2+\|s\|^2}}{\|x(T)\|}
\right)
\]
queries to their state preparation oracles.

Combining this overhead with the query complexity in ~\cref{prop:unstable_linear_solver}  {with precision $\min\{1,\frac{\|x(T)\|}{\|v(T)\|}\}\frac{\epsilon}{2}$}
 yields a quantum algorithm
that prepares an $\epsilon$-approximation of $\ket{x(T)}$ using
\begin{equation}
\begin{aligned}
       & \widetilde{\mathcal{O}}\left(\underbrace{\frac{\sqrt{\|x(T)-s\|^2+ \|s\|^2 } }{\|x(T)\|}}_{\text{Final state shift}}g_v\frac{\sum_{i = 0}^2\alpha_{F_{i,s}}\left\|F_{2,s}\right\|T^2}{ {\min\{\|v(T)\|,\|x(T)\|\}}\epsilon}\right)
\end{aligned}
\end{equation}
 queries to the block-encoding of $F_{i,s}$.

 Finally, substituting $\alpha_{F_{2,s}} = \alpha_{F_2}$, $\alpha_{F_{1,s}} = \alpha_{F_1}+2\|s\| \alpha_{F_2}$, $\alpha_{F_{0,s}} = \alpha_{F_0}+\alpha_{F_1} \|s\| + \alpha_{F_2}\|s\|^2$, and using $\|F_{2,s}\|=\|F_2\|$ and $s=x(0)$, the algorithm requires
 \begin{equation}
     \widetilde{\mathcal{O}}\left(\frac{\sqrt{\|x(T)-x(0)\|^2+ \|x(0)\|^2 } }{\|x(T)\|}\frac{\sum_{i = 0}^2\alpha_{F_{i}}(1+\|x(0)\|)^2\left\|F_{2}\right\|T^2g_v}{ {\min\{\|x(T)-x(0)\|, \|x(T)\|\}}\epsilon}\right),
 \end{equation}
 queries to the oracles for $F_i$ and $\ket{x(0)}$. This completes the proof.

\end{proof}

\section{Comparison with the pivot-switching method of~\cite{endo2024divergence}}

The pivot-switching (PS) method of~\cite{endo2024divergence} differs from our pivot-shifted construction in two key respects. First, in~\cite{endo2024divergence} the pivot $s$ is updated dynamically along the trajectory, whereas the present work uses a single, fixed pivot. Second, the lifted state in~\cite{endo2024divergence} is built from the original variable $x$ rather than the shifted variable $u$:
\begin{equation}
    \widehat{z}_{\mathrm{PS}} =
    \begin{bmatrix}
        x; x^{\otimes 2}; \cdots; x^{\otimes N}
    \end{bmatrix}^{\top}.
\end{equation}
Pivot switching can be viewed as a sequence of pivot shifts. Here, we only focus on a single pivot shift, corresponding to one step of pivot switching.

We start from the shifted quadratic representation, as in our PSC method,
\begin{equation}
    \partial_t u = F_{2,s}u^{\otimes 2} + F_{1,s}u + F_{0,s}, \qquad u=x-s.
\end{equation}
Equivalently, in the $x$ variable,
\begin{equation}
    \partial_t x = \widetilde{F}_{2,s}x^{\otimes 2} + \widetilde{F}_{1,s}x + \widetilde{F}_{0,s},
\end{equation}
where
\begin{equation}
\begin{aligned}
    \widetilde{F}_{2,s} &=
    F_{2,s}, \\
    \widetilde{F}_{1,s} &=  F_{1,s} - F_{2,s}\left(I\otimes s+s\otimes I\right), \\
    \widetilde{F}_{0,s} &= F_{0,s} - F_{1,s}s + F_{2,s}s^{\otimes 2}.
\end{aligned}
\end{equation}

For $j<N$, no truncation is applied because
\[
    \frac{d}{dt}x^{\otimes j}
\]
contains powers no higher than $x^{\otimes(j+1)}$, which is still included in
the order-$N$ lifted state. Thus the first $N-1$ block rows of the PS matrix
have the same quadratic Carleman form in the $x$ basis:
\begin{equation}
\begin{aligned}
    (B_{\mathrm{PS}}^{(N)})_{j,j+1}
    &= \widetilde{F}_{2,s}\otimes I^{\otimes(j-1)} + I\otimes \widetilde{F}_{2,s}\otimes I^{\otimes(j-2)} +
    \cdots
    + I^{\otimes(j-1)}\otimes \widetilde{F}_{2,s}, \\
    (B_{\mathrm{PS}}^{(N)})_{j,j}
    &= \widetilde{F}_{1,s}\otimes I^{\otimes(j-1)} + I\otimes \widetilde{F}_{1,s}\otimes I^{\otimes(j-2)} +
    \cdots
    + I^{\otimes(j-1)}\otimes \widetilde{F}_{1,s}, \\
    (B_{\mathrm{PS}}^{(N)})_{j,j-1}
    &= \widetilde{F}_{0,s}\otimes I^{\otimes(j-1)} + I\otimes \widetilde{F}_{0,s}\otimes I^{\otimes(j-2)} +
    \cdots
    + I^{\otimes(j-1)}\otimes \widetilde{F}_{0,s},
\end{aligned}
\end{equation}
for $1\le j<N$.

The only structural difference appears in the last block row. Before
truncation, the last lifted equation contains
\begin{equation}
    (B_{\mathrm{PS}}^{(N)})_{N,N+1}x^{\otimes(N+1)}.
\end{equation}
However, the PS method truncates in the shifted variable $u=x-s$, so the order
$N+1$ term in $u$ is removed. Equivalently,
\begin{equation}
    x^{\otimes(N+1)} \quad\longmapsto\quad x^{\otimes(N+1)}-(x-s)^{\otimes(N+1)}.
\end{equation}
Since
\begin{equation}
    (x-s)^{\otimes(N+1)} = x^{\otimes(N+1)} + \sum_{k=0}^{N} R_{N+1,k}(-s)x^{\otimes k},
\end{equation}
the highest-order term cancels and only lower-order terms remain:
\begin{equation}
    x^{\otimes(N+1)}-(x-s)^{\otimes(N+1)} = - \sum_{k=0}^{N} R_{N+1,k}(-s)x^{\otimes k}.
\end{equation}
Here $R_{N+1,k}(-s)$ denotes the coefficient matrix of $x^{\otimes k}$ in
the expansion of $(x-s)^{\otimes(N+1)}$.

Therefore, the PS matrix has the schematic form
\begin{equation}
B_{\mathrm{PS}}^{(N)}
=
\begin{bmatrix}
B_{1,1}^{x} & B_{1,2}^{x} & 0 & \cdots & 0 \\
B_{2,1}^{x} & B_{2,2}^{x} & B_{2,3}^{x} & \ddots & \vdots \\
0 & B_{3,2}^{x} & B_{3,3}^{x} & \ddots & 0 \\
\vdots & \ddots & \ddots & \ddots & B_{N-1,N}^{x} \\
C_1 & C_2 & C_3 & \cdots & C_N
\end{bmatrix},
\end{equation}
where the first $N-1$ rows are block tridiagonal, while the last row is dense.
The dense last-row blocks are
\begin{equation}
    C_k = \delta_{k,N-1}B_{N,N-1}^{x} + \delta_{k,N}B_{N,N}^{x} - B_{N,N+1}^{x}R_{N+1,k}(-s),
    \qquad 1\le k\le N.
\end{equation}
The corresponding affine term is
\begin{equation}
    d_{\mathrm{PS}}^{(N)} =
    \begin{bmatrix}
        \widetilde{F}_{0,s} \\
        0 \\
        \vdots \\
        -B_{N,N+1}^{x}R_{N+1,0}(-s)
    \end{bmatrix}.
\end{equation}

Moreover, PSC and PS are equivalent up to a change of coordinates. Indeed, the two lifted states are related by the finite translation map induced by $x=u+s$:
\begin{equation}
\widehat{z}_{\mathrm{PS}}=T_N(s)\widehat{z}_{\mathrm{PSC}},
\end{equation}
where \(T_N(s)\) expands each tensor power
\begin{equation}
    x^{\otimes j}= (u+s)^{\otimes j},
    \qquad
    1\le j\le N.
\end{equation}
The inverse map is obtained by replacing $s$ with $-s$, since $u=x-s$:
\begin{equation}
    \widehat{z}_{\mathrm{PSC}} = T_N(-s)\widehat{z}_{\mathrm{PS}},
    \qquad
    T_N(s)^{-1}=T_N(-s).
\end{equation}

Therefore, if the PSC system is written as
\begin{equation}
    \frac{d}{dt}\widehat{z}_{\mathrm{PSC}} = B_{\mathrm{PSC}}^{(N)}\widehat{z}_{\mathrm{PSC}} + d_{\mathrm{PSC}}^{(N)},
\end{equation}
then the same truncated dynamics expressed in the $x$ basis is
\begin{equation}
    \frac{d}{dt}\widehat{z}_{\mathrm{PS}} = T_N(s)B_{\mathrm{PSC}}^{(N)}T_N(-s)\widehat{z}_{\mathrm{PS}} + T_N(s)d_{\mathrm{PSC}}^{(N)}.
\end{equation}
Thus,
\begin{equation}
    B_{\mathrm{PS}}^{(N)} = T_N(s)B_{\mathrm{PSC}}^{(N)}T_N(-s),
    \qquad
    d_{\mathrm{PS}}^{(N)} = T_N(s)d_{\mathrm{PSC}}^{(N)}.
\end{equation}

This shows that the PS method does not define a different truncated trajectory. It rewrites the same shifted-order truncation in the unshifted \(x\)-coordinates. Consequently, when the lifted initial conditions are chosen consistently,
\begin{equation}
    \widehat{z}_{\mathrm{PS}}(0) = T_N(s)\widehat{z}_{\mathrm{PSC}}(0),
\end{equation}
the two methods satisfy
\begin{equation}
    \widehat{z}_{\mathrm{PS}}(t) = T_N(s)\widehat{z}_{\mathrm{PSC}}(t),
\end{equation}
and therefore recover the same physical solution $x(t)=u(t)+s$.

The practical difference is not the approximation itself, but the matrix representation: PSC remains sparse and block tridiagonal, whereas the PS construction generally produces a dense final block row after the coordinate transformation, which is inefficient for block encoding and quantum simulation. This dense row comes precisely from re-expressing the removed shifted term $u^{\otimes(N+1)}$ in the unshifted $x$ basis.

\end{document}